\documentclass[12pt]{article}
\usepackage{amsfonts,array}
\usepackage{multicol,subcaption,makecell}
\usepackage[font=footnotesize,labelfont=bf]{caption}
\usepackage{afterpage}
\usepackage{amsmath,bbm,dsfont,mathrsfs,mathtools,appendix}
\usepackage{amssymb}
\usepackage{graphicx}
\usepackage{fullpage}%
\usepackage{enumitem}
\usepackage{physics}
\usepackage{framed}
\usepackage[table]{xcolor}
\usepackage{authblk}

\usepackage[colorlinks=true,linkcolor=blue,citecolor=green,plainpages=false,pdfpagelabels]%
{hyperref}%
\hypersetup{
	colorlinks,
	linkcolor={blue!60!green},
	citecolor={green!50!yellow!75!black},
	urlcolor={blue!80!black},
	linktoc=all
}
\usepackage{cleveref} 
\providecommand{\U}[1]{\protect\rule{.1in}{.1in}}

\newtheorem{theorem}{Theorem}

\newtheorem{corollary}[theorem]{Corollary}

\newtheorem{definition}[theorem]{Definition}

\newtheorem{lemma}[theorem]{Lemma}

\newtheorem{proposition}[theorem]{Proposition}
	
\newtheorem{remark}[theorem]{Remark}

\newenvironment{proof}[1][Proof]{\noindent\textbf{#1.} }{\ \rule{0.5em}{0.5em}}

\newcommand{\B}[1]{\mathcal{L}({#1})}
\newcommand{\iden}{\mathbbm{1}}

\newcommand{\id}{{\rm{id}}} 

\newcommand{\R}{\mathbbm{R}}
\newcommand{\C}{\mathbb{C}}

\newcommand{\N}{\mathbb{N}}

\newcommand{\cA}{\mathcal{A}}
\newcommand{\cE}{\mathcal{E}}

\newcommand{\cM}{\mathcal{M}}
\newcommand{\cN}{\mathcal{N}}

\newcommand{\cD}{\mathcal{D}}

\newcommand{\E}{\mathbb{E}}

\def\>{{\rangle}}
\def\<{{\langle}}
\newcommand{\be}{\begin{equation}}
	\newcommand{\ee}{\end{equation}}
\newcommand{\bea}{\begin{eqnarray}}
	\newcommand{\eea}{\end{eqnarray}}

\newcommand{\eps}{\varepsilon}

\def\placeholder{\,\cdot\,}
\newcommand{\comment}[1]{}

\newcommand{\Bsa}[1]{\mathcal{L}^{\operatorname{sa}}(#1)}

\newcommand{\ID}{\operatorname{ID}}

\numberwithin{equation}{section}
\numberwithin{theorem}{section}

\usepackage{color}
\definecolor{colorthree}{rgb}{0.01,0.51,0.93}

\usepackage{setspace}

\usepackage{newunicodechar}
\newunicodechar{ρ}{\rho}
\newunicodechar{β}{\beta}
\newunicodechar{σ}{\sigma}
\newunicodechar{λ}{\lambda}
\newunicodechar{Λ}{\Lambda}
\newunicodechar{μ}{\mu}
\newunicodechar{ν}{\nu}
\newunicodechar{ψ}{\psi}
\newunicodechar{ϕ}{\varphi}
\newunicodechar{Φ}{\cN}
\newunicodechar{φ}{\cN}
\newunicodechar{π}{\pi}
\newunicodechar{α}{\alpha}
\newunicodechar{ϵ}{\epsilon}
\newunicodechar{ε}{\varepsilon}
\newunicodechar{δ}{\delta}
\newunicodechar{ω}{\omega}
\newunicodechar{Δ}{\Delta}
\newunicodechar{Σ}{\Sigma}

\title{Gaussian mean width strong converse bound on the classical identification capacity of quantum channels}
\author[1]{Satvik Singh}

\affil[1]{Department of Mathematics, Technical University of Munich, Garching, Germany}

\affil[1]{Munich Center for Quantum Science and Technology (MCQST), Munich, Germany}

\date{}

\begin{document}

\maketitle

\begin{abstract}
We establish a single-letter and efficiently computable strong converse bound on the classical identification capacity of quantum channels. By equipping the $n$-fold channel output space with a product state-weighted $\sigma$-Euclidean geometry, we allow trace-distance separation constraints for identification codes to be controlled by Euclidean covering estimates. Using Sudakov's inequality, we bound the covering numbers of the $n$-fold channel outputs via their Gaussian mean widths in the weighted geometry, whose exponential growth in $n$ is governed by the operator norm of a single-letter positive operator. Upon optimizing over all weighing states $\sigma$, this yields a strong converse bound on the identification capacity of the channel, which also admits a semidefinite representation. Our method improves the best known converse bounds on the identification capacity of several important examples, such as depolarizing, Pauli, erasure, and amplitude damping channels. We also discuss extensions of this method to more general Euclidean geometries on the output space.
\end{abstract}

\tableofcontents

\section{Introduction}

A basic geometric object associated with a quantum channel $\cN:\B{A}\to\B{B}$ is the image of the set of $n$-partite density operators $\cD (A^{\otimes n})\subseteq \Bsa{A^{\otimes n}}$ under the $n$-fold independent and identical application of the channel:
\begin{equation}\label{eq:N-image-intro}
I_n(\cN):=\cN^{\otimes n} (\cD(A^{\otimes n})) = \{\cN^{\otimes n}(\rho) : \rho \in \cD(A^{\otimes n}) \} \subseteq \Bsa{B^{\otimes n}},    
\end{equation}
where $\B{A}$ denotes the space of all linear operators acting on a complex Hilbert space\footnote{We only consider finite-dimensional Hilbert spaces in this paper.} $A$, and $\Bsa{A}$ denotes the (real) space of self-adjoint operators in $\B{A}$. Understanding the metric complexity \cite{Avidan2015analysis} of the sets $I_n(\cN)$ is of interest in quantum information theory, as it is directly relevant to coding problems, most notably \emph{identification} via quantum channels \cite{Ahlswede1989ID, Lober1999thesis-ID, Winter2013survey-ID}. Indeed, any identification code gives rise to a large separated subset (i.e. a packing) of the channel image $I_n(\cN)$ in trace distance (c.f. Lemma~\ref{lemma:CID-separation}). Thus, upper bounds on covering numbers of $I_n(\cN)$ in trace distance yield converse bounds for identification. 

In this paper, we use the following
weighted Euclidean inner products on $\Bsa{B^{\otimes n}}$:
\begin{equation}
    \langle Y,Z\rangle_{\sigma^{\otimes n}}
    :=
    \Tr\!\left(
        Y
        (\sigma^{\otimes n})^{-1/2}
        Z (\sigma^{\otimes n})^{-1/2}
    \right), \qquad Y,Z\in \Bsa{B^{\otimes n}},
\end{equation}
where $\sigma\in\cD_+(B)$ is a full-rank state. The corresponding norm dominates the trace norm (c.f. Lemma~\ref{lemma:weighted-norm-dominates-trace}): $\norm{Y}_1\leq \norm{Y}_{\sigma^{\otimes n}}$. Consequently, covering the channel image \(I_n(\cN)\) in the weighted geometry gives a covering
strong enough for identification. 

The main geometric tool we use to control the covering numbers is the Gaussian mean width. For any subset $S\subseteq \R^m$ in a Euclidean space, its \emph{Gaussian mean width} is defined as
\begin{equation}
    w_G(S) := \E \sup_{x\in S} \langle g, x \rangle,
\end{equation}
where $g\sim N(0,1_m)$ is a standard Gaussian random vector in $\R^m$ \cite{Avidan2015analysis, Vershynin2018HDP}. The Gaussian width quantifies how wide a set `looks' in random directions on average, and is a fundamental geometric measure of a set. Moreover, via the Sudakov \cite{Sudakov1971} and Dudley \cite{Dudley1967} inequalities, it serves as a good proxy for the metric complexity of $S$, in the following sense:
\begin{equation}\label{eq:sudakov-dudley-intro}
 c \sup_{\eps>0} \eps \sqrt{\log C_{\eps}(S)} \leq  w_G(S) \leq C \int_0^{\infty} \sqrt{\log C_{\eps}(S)} d\eps ,
\end{equation}
where $C_{\eps}(S)$ denotes the smallest number of Euclidean balls with radius $\eps>0$ and centers in $S$ that cover $S$ (c.f. Definition~\ref{def:packing-covering}), and $c,C>0$ are absolute constants (see \cite[Chapter 4]{Avidan2015analysis} or \cite[Chapters 7,8]{Vershynin2018HDP}). 

Our central observation is that suitable control over the Gaussian widths $w_{G,\sigma^{\otimes n}}(I_n(\cN))$, combined with the domination bound $\norm{\cdot}_1 \leq \norm{\cdot}_{\sigma^{\otimes n}}$ and Sudakov inequality, gives a non-trivial method to control the covering numbers of $I_n(\cN)$ in trace distance. In order to control the exponential growth of $w_{G,\sigma^{\otimes n}}(I_n(\cN))$, we introduce the \emph{weighted singular operator} $Q_{\cN,\sigma}$ associated with the channel $\cN$ and the state
$\sigma$. More precisely, \(Q_{\cN,\sigma}\)
is the second moment of a random Gaussian direction $G_{\sigma}\in \Bsa{B}$, sampled according to the
$\sigma$-weighted Euclidean geometry, after being pulled back through the $\sigma$-adjoint $\cN^{*,\sigma}: \Bsa{B}\to \Bsa{A}$:
\begin{equation}\label{eq:weighted-singular-Q-intro}
    Q_{\cN,\sigma}
    :=
    \E\left[
        \left(\cN^{*,\sigma}(G_\sigma)\right)^2
    \right]
    \in\Bsa{A}.
\end{equation}
We prove that the operator
norm of $Q_{\cN, \sigma}$ governs the asymptotic growth rate of the widths $w_{G,\sigma^{\otimes n}}(I_n(\cN))$; see the result stated below and Theorem~\ref{theorem:weighted-gaussian-width-growth} for more details.

\begin{theorem} \label{theorem:main-intro}
Let $\cN:\B{A}\to\B{B}$ be a quantum channel and $\sigma\in \cD_+(B)$ be a full-rank density operator. Then,
\begin{equation}
\forall n\in \N : \quad   w_{G,\sigma^{\otimes n}}\bigl(I_n(\cN)\bigr)
    \leq
    \sqrt{2n\log d_A}
    \norm{Q_{\cN,\sigma}}_\infty^{n/2} .
\end{equation}
Moreover, if $\sigma=\cN(\eta)$ for some full rank input density $\eta\in \cD_+(A)$, then,
\begin{equation}
\forall n\in \N: \quad\frac{1}{\sqrt{2\pi}}
    \sqrt{\norm{Q_{\cN,\sigma}}_\infty^n-1}
    \leq
    w_{G,\sigma^{\otimes n}}\bigl(I_n(\cN)\bigr)
    \leq
    \sqrt{2n\log d_A
    \left(\norm{Q_{\cN,\sigma}}_\infty^n-1 \right)}.
\end{equation}
\end{theorem}

We use this estimate, along with the domination bound $\norm{\cdot}_1 \leq \norm{\cdot}_{\sigma^{\otimes n}}$, Sudakov inequality, and trace norm identification constraints (Lemma~\ref{lemma:CID-separation}) to obtain the following strong converse bound on the classical
identification capacity of quantum channels (see Theorem~\ref{theorem:weighted-gaussian-converse}).

\begin{theorem} \label{thm:gaussian-converse-intro}
Let $\cN:\B{A}\to \B{B}$ be a quantum channel. Then, for every $\lambda_1,\lambda_2>0$ with $\lambda_1+\lambda_2<1$, the size $N$ of the largest $(n,N,\lambda_1, \lambda_2)$ identification code for $\cN$ (c.f. Definition~\ref{def:IDcodes}), denoted $N_{(n,\lambda_1,\lambda_2)}(\cN)$, satisfies
\begin{align}\label{eq:gaussian-converse-intro}
    \limsup_{n\to\infty}
    \frac{1}{n}
    \log\log N_{(n,\lambda_1,\lambda_2)}(\cN)
    \leq
    \inf_{\sigma\in\cD_+(B)} \biggl(
    \log\norm{Q_{\cN,\sigma}}_\infty \biggr),
\end{align}
where the optimization on the right hand side is over all full-rank densities $\sigma\in \cD_+(B)$. Furthermore, this optimization admits a semidefinite representation (c.f. Eq.~\eqref{eq:converse-SDP}).
\end{theorem} 

Our result improves the best known converse bounds on the identification capacity of important examples, such as depolarizing, Pauli, erasure, and amplitude damping channels (see Section~\ref{sec:application}). We illustrate this improvement for the qubit depolarizing channel in Figure~\ref{fig:qubit-depolarizing-intro} (see Section~\ref{sec:depol} for details). We also discuss further extensions of our method for more general Euclidean geometries on the output space in Section~\ref{sec:discussion}. 

\begin{figure}[h]
    \centering
    \includegraphics[width=0.69\linewidth]{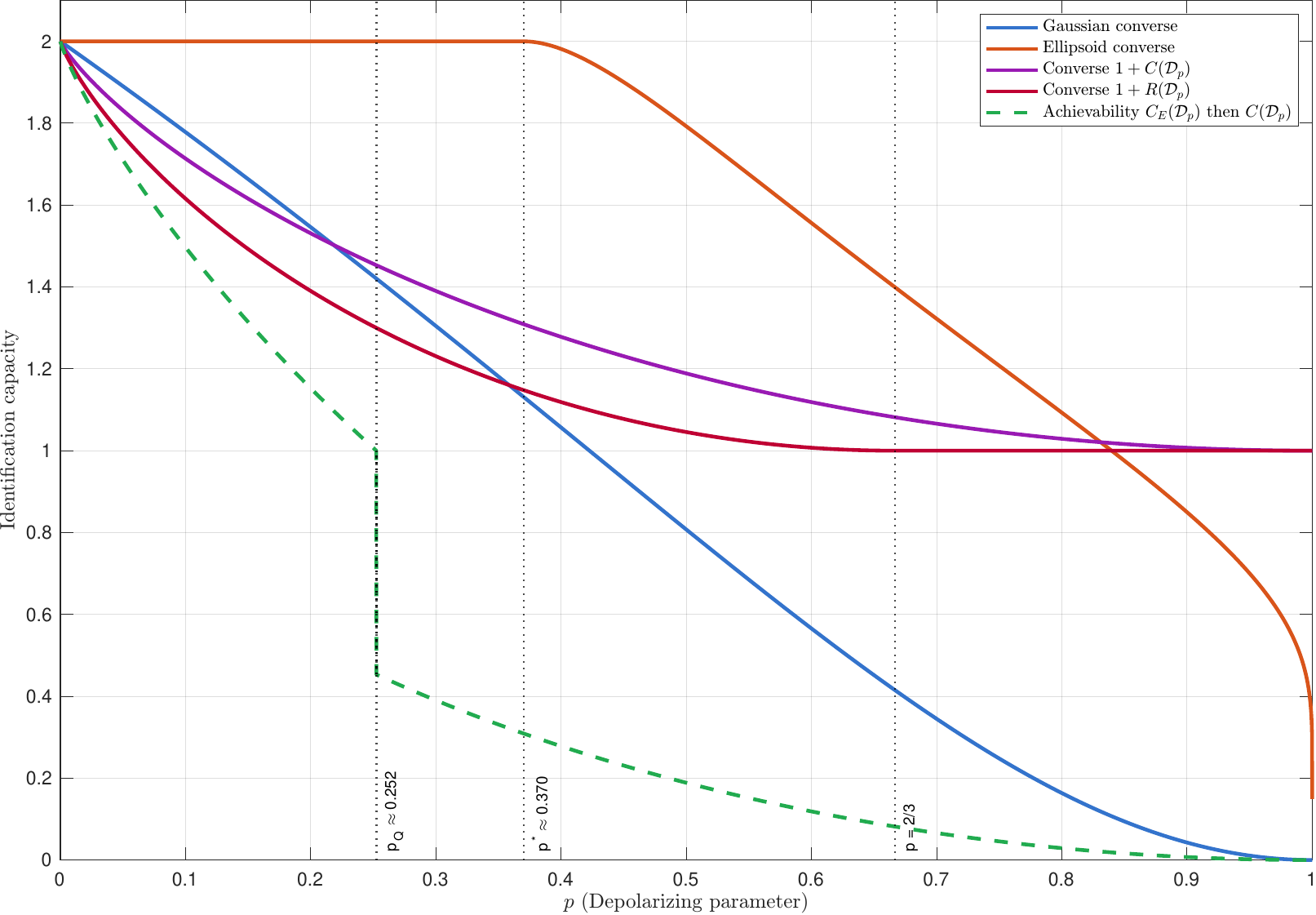}
    \caption{Strong converse and achievability bounds on the classical identification capacity of the qubit depolarizing channel $\cD_p$ (see Eq.~\eqref{eq:qubit-depol}). The blue curve shows the Gaussian converse bound from Theorem~\ref{thm:gaussian-converse-intro} (see Eq.\eqref{eq:qubit-depolarizing-gaussian}). The orange curve shows the Ellipsoid converse bound from \cite{ye2026strongconverseboundsclassical} (see Eq.~\eqref{eq:qubit-depolarizing-ellipsoid}). The purple and red curves show the classical capacity \cite{ye2026strongconverseboundsclassical} and quantum capacity \cite{Atif2024CIDstrongconverse} converse bounds (see Eqs.~\eqref{eq:CID(Dp)<1+C(Dp)} and \eqref{eq:CID(Dp)<1+Q(Dp)}), respectively. The dashed curve shows the achievability bound from \cite{Hayden2012QID-achievability, Winter2013survey-ID} (see Eq.~\eqref{eq:qubit-depolarizing-achievability}).}
    \label{fig:qubit-depolarizing-intro}
\end{figure}

\subsection{Outline}

We briefly describe the layout of the paper below. 
\begin{itemize}
    \item In Section~\ref{sec:prelim}, we gather some preliminary material for the rest of the paper.
    \item In Section~\ref{sec:main}, we present the main results on the asymptotic behavior of the weighted Gaussian mean widths $w_{G,\sigma^{\otimes n}}(I_n(\cN))$ (c.f. Theorem~\ref{theorem:weighted-gaussian-width-growth}) and the corresponding Gaussian strong converse bound on identification (c.f. Theorem~\ref{theorem:weighted-gaussian-converse}).
    \item In Section~\ref{sec:application}, we apply our results to several important quantum channels, namely depolarizing, Pauli, erasure, and amplitude damping channels. 
    \item In Section~\ref{sec:discussion}, we conclude by discussing possible extensions and ideas for future work.
\end{itemize}

\section{Preliminaries} \label{sec:prelim}

\subsection{Basic notation}

We denote the cardinality of a set $S$ by $\abs{S}$. For $n\in\N$, we denote $[n]:=\{0,1,2,\ldots,n-1\}$. The base-$2$ and natural logarithms are denoted by $\log$ and $\ln$, respectively. For a
random variable $X$, we denote its expectation by $\E X$. 

We denote complex Hilbert spaces by capital letters $A,B,C$, with corresponding dimensions
$d_A,d_B$ and $d_C$, respectively. $\B{A}$ denotes the space of all linear operators acting
on $A$. The trace-norm, Schatten $\alpha$-norm for $\alpha\in(1,\infty)$, and the operator
norm on $\B{A}$ are denoted by $\norm{\cdot}_1$, $\norm{\cdot}_\alpha$, and
$\norm{\cdot}_\infty$, respectively. The real space of Hermitian operators in $\B{A}$ is denoted by $\Bsa{A}$. The maximum eigenvalue of a Hermitian operator $X\in\Bsa{A}$ is denoted by
$\lambda_{\max}(X)$. For $X,Y\in\Bsa{A}$, $X\geq Y$ means that $X-Y$ is positive
semi-definite. The convex set of quantum states or density operators, i.e. positive semi-definite operators in
$\Bsa{A}$ with unit trace, is denoted by $\cD(A)$. The set of full-rank states in $\cD(A)$ is denoted by $\cD_+(A)$. The identity operator in $\B{A}$ is
denoted by $\iden_A$. A linear map $\cN:\B{A}\to\B{B}$ is called a quantum channel if it is completely positive
and trace-preserving.

\subsection{Weighted Euclidean geometries}

The Hilbert-Schmidt inner product on $\Bsa{A}$ is denoted by $\langle X,Y\rangle := \Tr(X Y),$ with the corresponding norm $\norm{X}_2 :=
\sqrt{\langle X,X\rangle}.$ We will also use weighted Euclidean structures on $\Bsa{B}$. Let $\sigma\in\cD_+(B)$ be full rank. For
$Y,Z\in\Bsa{B}$, define the $\sigma$\emph{-weighted inner product}
\begin{equation}
    \langle Y,Z\rangle_\sigma
    :=
    \Tr\!\left(Y\sigma^{-1/2}Z\sigma^{-1/2}\right),
\end{equation}
with the corresponding norm $\norm{Y}_\sigma
    :=
    \sqrt{\langle Y,Y\rangle_\sigma}.$ Equivalently, we can define positive definite\footnote{Not to be confused with positivity-preserving maps. Here, $\Tr(Y W_{\sigma}(Y))>0$ for all $\Bsa{B} \ni Y\neq 0$, which means $W_{\sigma}$ is (Hilbert-Schmidt) positive definite as an operator acting on $\Bsa{B}$.} linear maps $\Phi_\sigma, W_{\sigma}:\Bsa{B}\to\Bsa{B}$ as $\Phi_\sigma(Y):=\sigma^{-1/4}Y\sigma^{-1/4}, W_{\sigma}=\Phi_{\sigma}^2$, so that
\begin{equation}\label{eq:Wsigma}
    \langle Y,Z\rangle_\sigma
    =
    \langle \Phi_\sigma(Y),\Phi_\sigma(Z)\rangle = \Tr (YW_{\sigma}(Z) ),
    \qquad
    \norm{Y}_\sigma
    =
    \norm{\Phi_\sigma(Y)}_2.
\end{equation}
Thus, $\Phi_\sigma$ is an isometry from
$(\Bsa{B},\langle\cdot,\cdot\rangle_\sigma)$ to
$(\Bsa{B},\langle\cdot,\cdot\rangle)$.

The following lemma shows that the weighted norm dominates the trace norm.
\begin{lemma}\label{lemma:weighted-norm-dominates-trace}
    For every full-rank state $\sigma\in \cD_+(B)$ and $Y\in \Bsa{B}$, $\norm{Y}_1\leq \norm{Y}_\sigma$.
\end{lemma}
\begin{proof}
    By the matrix H\"older's inequality \cite[Chapter IV]{Bhatia1997matrix}, we get
\begin{align}
    \norm{Y}_1
    =
    \norm{\sigma^{1/4}\sigma^{-1/4}Y\sigma^{-1/4}\sigma^{1/4}}_1
    &\leq
    \norm{\sigma^{1/4}}_4\,\norm{\sigma^{-1/4}Y\sigma^{-1/4}}_2\,\norm{\sigma^{1/4}}_4 \\
    &=
    \norm{\sigma^{-1/4}Y\sigma^{-1/4}}_2
    =
    \norm{Y}_\sigma,
\end{align}
where we used $\norm{\sigma^{1/4}}_4=(\Tr\sigma)^{1/4}=1$.
\end{proof}

Let $\cN:\B{A}\to\B{B}$ be a quantum channel. If $\sigma\in\cD_+(B)$ is full rank, we can view $\cN$ as a real linear map $\cN:
    (\Bsa{A},\langle\cdot,\cdot\rangle)
    \to
    (\Bsa{B},\langle\cdot,\cdot\rangle_\sigma)$, since $\cN$ is Hermiticity-preserving. The corresponding weighted adjoint $\cN^{*,\sigma}:\Bsa{B}\to\Bsa{A}$ is defined via the relation
\begin{equation}
    \langle Y,\cN(X)\rangle_\sigma
    =
    \langle \cN^{*,\sigma}(Y),X\rangle,
    \qquad
    X\in\Bsa{A},\; Y\in\Bsa{B}.
\end{equation}
Explicitly, it is easy to check that $\cN^{*,\sigma}(Y) =
    \cN^*\!\left(\sigma^{-1/2}Y\sigma^{-1/2}\right)$, where $\cN^*$ is the regular Hilbert-Schmidt adjoint of $\cN:(\Bsa{A},\langle\cdot,\cdot\rangle)
    \to
    (\Bsa{B},\langle\cdot,\cdot\rangle)$.

Next, we recall the singular value decomposition. Let
\(\cN:\B{A}\to\B{B}\) be a quantum channel and \(\sigma\in\cD_+(B)\). Then, the \(\sigma\)-\emph{weighted singular values} \(s_a:=s_a(\cN,\sigma)\) are the singular values of $\cN:
    (\Bsa{A},\langle\cdot,\cdot\rangle)
    \to
    (\Bsa{B},\langle\cdot,\cdot\rangle_\sigma)$, ordered in non-increasing fashion:
\begin{equation}
    s_0\geq \ldots \geq s_{r-1}>0=s_r=\ldots=s_{d_A^2-1},
    \qquad
    r=\rank \cN.
\end{equation}
Equivalently, the positive operator \(\cN^{*,\sigma}\cN\) on \(\Bsa{A}\) has eigenvalues
\(s_a(\cN, \sigma)^2\). Thus, we can choose a Hilbert-Schmidt orthonormal basis
\(\{F_a\}_a\subseteq\Bsa{A}\) and a \(\sigma\)-orthonormal family
\(\{G_a\}_{a:s_a>0}\subseteq\Bsa{B}\), meaning $\langle F_a, F_b \rangle = \delta_{ab}$ and $\langle G_a,G_b\rangle_\sigma=\delta_{ab},$ such that
\begin{align}
    \cN(F_a) &= s_a G_a, \qquad a=0,1,\ldots,r-1, \label{eq:weighted-SVD} \\
    \cN(F_a) &= 0, \qquad\quad\,\,\,\, a\geq r \label{eq:weighted-SVD2}
\end{align}
Equivalently, $H_a:=\Phi_\sigma(G_a)$ are Hilbert-Schmidt orthonormal and $\Phi_\sigma\cN(F_a)=s_aH_a$. Thus, $s_a(\cN, \sigma)$ are precisely the ordinary
Hilbert-Schmidt singular values of $\Phi_\sigma \circ\cN$.

\subsection{Packings and coverings}

\begin{definition}\label{def:packing-covering}
Let $(X,d)$ be a metric space, let $S\subseteq X$ be a bounded subset. The closed ball of radius $\eps>0$ around $x\in X$ is denoted by $B_\eps(x) := \{ y\in X : d(x,y)\leq \eps \}$.
\begin{itemize}
    \item A (finite) subset $\mathscr{P}\subseteq S$ is called an $\eps$\emph{-packing} of $S$ if $\forall x\neq y\in \mathscr{P}, \,d(x,y)>\eps$. The \emph{$\eps$-packing number} of $S$ with respect to the metric $d$ is defined by
\begin{equation}
    P_{\eps}(S;d) := \max \left\{ |\mathscr{P}| : \mathscr{P}\subseteq S \, \text{ is an } \eps\text{-packing} \right\}.
\end{equation} 
    \item A (finite) subset $\mathscr{C}\subseteq S$ is called an $\eps$\emph{-covering} of $S$ if $S \subseteq \bigcup_{x\in \mathscr{C}} B_\eps(x)$. The \emph{$\eps$-covering number} of $S$ with respect to the metric $d$ is defined by
\begin{equation}
    C_{\eps}(S;d) := \min \left\{ |\mathscr{C}| : \mathscr{C}\subseteq S \, \text{ is an } \eps\text{-covering} \right\}.
\end{equation} 
\end{itemize}
\end{definition}

We have the following basic relations between coverings and packings for arbitrary subsets $S\subseteq X$ in metric spaces. First, if $\mathscr{P}\subseteq S$ is a $2\eps$-packing, and $\mathscr{C}\subseteq S$ is an $\eps$-covering, then since $S \subseteq \bigcup_{x\in \mathscr{C}} B_\eps(x)$, each ball $B_{\eps}(x)$ with $x\in \mathscr{C}$ can contain at most one element of $\mathscr{P}$. Secondly, a maximal $\eps$-packing of $S$ (with respect to inclusion) is an $\eps$-covering. Hence,
    \begin{equation}\label{eq:cover-pack-inequality}
        P_{2\eps}(S;d) \leq C_{\eps}(S;d) \leq P_{\eps}(S;d).
    \end{equation}

In Euclidean spaces, Sudakov and Dudley inequalities are powerful tools to estimate packing/covering numbers of bounded sets in terms of their Gaussian mean widths (see e.g. \cite{Ledoux1991prob-banach, Avidan2015analysis, Vershynin2018HDP} for detailed expositions). In what follows, we denote by $d(x,y):= \sqrt{\sum_i (x_i-y_i)^2}$ the standard Euclidean distance in $\R^m$.

\begin{definition}\label{def:gaussian-width}
The \emph{Gaussian mean width} of a bounded set $S\subseteq \mathbb R^m$ is defined as
\begin{equation}
    w_G(S):=\mathbb E \sup_{x\in S}\langle g,x\rangle,
\end{equation}
where $g\sim N(0,1_m)$ is a standard Gaussian random vector\footnote{i.e. $g$ is a $\R^m$-valued random variable whose coordinates (in any orthonormal basis) are independent and identically distributed (i.i.d) according to the standard Gaussian distribution $N(0,1)$.} in $\mathbb R^m$.
\end{definition}

\begin{lemma}[Sudakov inequality in $\R^m$]  \label{lemma:sudakov} \cite{Sudakov1971} \cite[Chapter 3]{Ledoux1991prob-banach} \cite[Chapter 7]{Vershynin2018HDP} \\
There exists a universal constant $C>0$ such that for every bounded set
$S\subseteq \mathbb R^m$ and $\eps>0$,
\begin{equation}
\log C_{\eps}(S;d)\leq C\,\frac{w_G(S)^2}{\eps^2}.
\end{equation} 
\end{lemma}

\begin{lemma}[Dudley inequality in $\R^m$] \label{lemma:dudley} \cite{Dudley1967} \cite[Chapter 11]{Ledoux1991prob-banach} \cite[Chapter 12]{Vershynin2018HDP} \\
There exists a universal constant $c>0$ such that for every bounded set
$S\subseteq \mathbb R^m$,
\begin{equation}
w_G(S) \leq c \int_0^{\infty} \sqrt{\log C_{\eps}(S;d)} \,\,  d\eps.
\end{equation}
\end{lemma}

It turns out that Sudakov and Dudley inequalities are sharp up to logarithmic factors in the dimension. More precisely, for any bounded set $S\subseteq \R^m$, the following holds:
\begin{equation}\label{eq:sudakov-dudley}
 c \sup_{\eps>0} \eps \sqrt{\log C_{\eps}(S;d)} \leq  w_G(S) \leq C \log (m) \sup_{\eps>0}\eps \sqrt{\log C_{\eps}(S;d)},
\end{equation}
where $c,C>0$ are absolute constants \cite[Theorem 8.1.13]{Vershynin2018HDP}.

\subsection{Classical identification task}\label{sec:ID}

Let $\cN:\B{A}\to \B{B}$ be a noisy channel shared between Alice and Bob. 

In the task of classical information \emph{transmission} from Alice to Bob, Alice first encodes a message $i\in \{0,1,\ldots ,N-1\}$ from a list of possible messages  in a quantum state $\rho_i\in \cD(A)$, sends it via $\cN$ to Bob, who performs a measurement $\{\Lambda_i\}_{i\in [N]}\subseteq \B{B}$, $\Lambda_i\geq 0$, $\sum_i \Lambda_i = \iden_B$ in an attempt to decode exactly which message Alice intended to send. It is known from the original work of Shannon \cite{Shannon1948, Cover2005book, Wilde2016} that the maximum number of messages that can be reliably transmitted via $n$ independent uses of $\cN$ scales \emph{exponentially} with $n$. This motivates the definition of transmission capacity.

\begin{definition}\label{def:Transmission-codes}
    A $(n,N,\lambda)$ (classical) \emph{transmission code} for a channel $\cN: \B{A}\to \B{B}$ is defined by encoding quantum states $\{\rho_i\}_{i\in [N]} \subseteq \cD (A^{\otimes n})$ and a decoding measurement $\{\Lambda_i\}_{i\in [N]} \subseteq \B{B^{\otimes n}}$, $\Lambda_i\geq 0$, $\sum_i \Lambda_i = \iden_{B^{\otimes n}}$, such that
    \begin{align}
        \forall i \in [N] : \qquad \Tr ( \cN^{\otimes n}(\rho_i) \Lambda_i ) &\geq 1-\lambda .
    \end{align}
    For a given $n\in \N, \lambda\in [0,1)$, the maximum size of all $(n,N,\lambda)$ codes for $\cN$ is 
    \begin{equation}
        N_{(n, \lambda)}(\cN) := \max \{N : \exists (n,N, \lambda) \text{ classical transmission code for } \cN \},
    \end{equation}
    and the classical \emph{transmission capacity} is defined as
    \begin{align}
        C(\cN) &:= \inf_{\lambda >0} \liminf_{n\to \infty} \frac{1}{n} \log N_{(n, \lambda)}(\cN) .
    \end{align}
\end{definition}

If, instead, Bob is only interested in \emph{identifying} whether the sent message $i$ is equal to a fixed message $j$ or not, we can do much better. Indeed, the maximum number of messages that can be reliably identified via $n$ independent uses of $\cN$ scales \emph{doubly exponentially} with $n$ \cite{Ahlswede1989ID, Lober1999thesis-ID, Ahlswede2002strong-ID}. This motivates the definition of identification capacity.

\begin{definition}\label{def:IDcodes}
    A $(n,N,\lambda_1,\lambda_2)$ (classical) \emph{identification (ID) code} for a quantum channel $\cN: \B{A}\to \B{B}$ is defined by pairs $\{(\rho_i, D_i) \}_{i\in [N]}$ of encoding states $\rho_i \in \cD (A^{\otimes n})$ and decoding effects $D_i\in \B{B^{\otimes n}}$, $0_{B^{\otimes n}} \leq D_i \leq \iden_{B^{\otimes n}}$, such that
    \begin{align}
        \forall i \in [N] : \qquad \Tr ( \cN^{\otimes n}(\rho_i) D_i ) &\geq 1-\lambda_1, \\ 
        \forall j\neq i\in [N] : \qquad \Tr (\cN^{\otimes n}(\rho_i) D_j) &\leq \lambda_2  .
    \end{align}
    For a given $n\in \N, \lambda_1,\lambda_2\in [0,1)$, the maximum size of all $(n,N,\lambda_1, \lambda_2)$ codes for $\cN$ is\footnote{To avoid trivialities, we will always assume that $\lambda_1,\lambda_2>0$ and $\lambda_1+\lambda_2<1$.} 
    \begin{equation}
        N_{(n, \lambda_1, \lambda_2)}(\cN) := \max \{N : \exists (n,N, \lambda_1, \lambda_2) \text{ classical ID code for } \cN \},
    \end{equation}
    and the classical \emph{ID capacity} is defined as\footnote{Whenever a double log of an integer code size appears, we use the
convention $\log\log 1 := 0.$ Equivalently, we may replace \(\log\log N\) by \(\log\log\max\{N,2\}\) without changing the asymptotic rates.}
    \begin{align}
        C_{\ID}(\cN) &:= \inf_{\lambda_1, \lambda_2 >0} \liminf_{n\to \infty} \frac{1}{n} \log \log N_{(n, \lambda_1, \lambda_2)}(\cN) .
    \end{align}
\end{definition}

The basic idea behind proving converse bounds for identification is the following well-known elementary observation. Let $\{(\rho_i,D_i)\}_{i\in [N]}$ be an $(n,N,\lambda_1,\lambda_2)$ ID code for $\cN$. Set $\omega_i := \cN^{\otimes n}(\rho_i)$. By the defining properties of the code and the variational characterization of the trace norm, we get for $i\neq j$:
\begin{align}
    d_{\Tr}(\omega_i,\omega_j)
    :=
    \frac{1}{2}\norm{\omega_i-\omega_j}_1
    &=
    \max_{0\leq D \leq \iden_{B^{\otimes n}}} \Tr\bigl((\omega_i-\omega_j)D\bigr) \\
    &\geq
    \Tr\bigl((\omega_i-\omega_j)D_i\bigr) \\
    &\geq 1- \lambda_1-\lambda_2. \label{eq:CID-separation}
\end{align}
Hence, for $\zeta<1-\lambda_1-\lambda_2$, the output states $\{\omega_i\}_{i\in [N]}$ of any $(n,N,\lambda_1,\lambda_2)$ ID code for $\cN$ form an $\zeta$-packing of the image set $I_n(\cN):= \cN (\cD(A^{\otimes n})) \subseteq \Bsa{B^{\otimes n}}$ with respect to the trace distance. We note this observation in the following lemma.

\begin{lemma}\label{lemma:CID-separation}
For a channel $\cN:\B{A}\to \B{B}$, define the $n$-shot output image set
\begin{equation}
    I_n(\cN) := \cN^{\otimes n}( \cD(A^{\otimes n}) ) = \{\cN^{\otimes n}(\rho) : \rho \in \cD(A^{\otimes n}) \} \subseteq \cD(B^{\otimes n}).
\end{equation}
Then, for every $\lambda_1,\lambda_2>0$ with $\lambda_1+\lambda_2<1$, and every  $0<\zeta< 1-\lambda_1-\lambda_2$, we have
\begin{equation}
    N_{(n,\lambda_1,\lambda_2)}(\cN)
    \leq P_{\zeta}(I_n (\cN) ; d_{\Tr}) \leq 
    C_{\zeta/2}\bigl(I_n(\cN); d_{\Tr}).
\end{equation}
\end{lemma}

\begin{proof}
The proof follows directly from Eqs.~\eqref{eq:CID-separation} and \eqref{eq:cover-pack-inequality}.
\end{proof}

We refer the reader to the review \cite{Winter2013survey-ID} for a more detailed discussion on identification.

\section{Main results} \label{sec:main}

\subsection{The singular operator}

Let $\cN:\B{A}\to\B{B}$ be a quantum channel and let
$\sigma\in\cD_+(B)$ be a full rank state. We can view $\cN$ as a linear map $\cN:(\Bsa{A},\langle\cdot,\cdot\rangle)
    \to
    (\Bsa{B},\langle\cdot,\cdot\rangle_\sigma)$ with adjoint $
    \cN^{*,\sigma}:\Bsa{B}\to\Bsa{A}.$ Let $G_{\sigma}$ be a standard Gaussian vector in the Euclidean space
$(\Bsa{B},\langle\cdot,\cdot\rangle_\sigma)$. We define the
$\sigma$-\emph{weighted singular operator} of $\cN$ as follows:
\begin{equation}\label{eq:weighted-singular-Q}
    Q_{\cN,\sigma}
    :=
    \E\left[
        \left(\cN^{*,\sigma}(G_\sigma)\right)^2
    \right]
    \in\Bsa{A}.
\end{equation}

Clearly, $Q_{\cN,\sigma}$ is positive semi-definite. The following lemma illustrates how this definition relates to the $\sigma$-weighted singular value decompositions of $\cN$.

\begin{lemma}\label{lemma:weighted-Q-SVD}
Let $\cN:\B{A}\to \B{B}$ be a quantum channel and $\sigma\in\cD_+(B)$ be a full rank state. Choose any $\sigma$-orthonormal basis $\{G_a\}_{a}\subseteq\Bsa{B}$. Then,
\begin{equation}
    Q_{\cN,\sigma}
    = \sum_a  \left( \cN^{*,\sigma}(G_a) \right)^2 .
\end{equation}
Hence, if $\cN(F_a)=s_aG_a$
is a $\sigma$-weighted singular value decomposition of $\cN$, where $s_a:=s_a(\cN,\sigma)$ are the $\sigma$-weighted singular values, $\{F_a\}_a\subseteq\Bsa{A}$ is Hilbert-Schmidt orthonormal and $\{G_a\}_{a}\subseteq\Bsa{B}$ is $\sigma$-orthonormal, we get
\begin{equation}
    Q_{\cN,\sigma} = \sum_a  \left( \cN^{*,\sigma}(G_a) \right)^2 
    =
    \sum_a s_a^2F_a^2 .
\end{equation}
\end{lemma}

\begin{proof}
Consider a standard random Gaussian vector $G=
    \sum_{a} g_aG_a$ in the weighted Euclidean space $(\Bsa{B},\langle\cdot,\cdot\rangle_\sigma)$, where $g_a\sim N(0,1)$ are i.i.d. standard real Gaussians, and $\{G_a\}_{a}\subseteq\Bsa{B}$ is a $\sigma$-orthonormal basis. Then,
\begin{align}
    Q_{\cN,\sigma}
    = \E\left[
        \left(\cN^{*,\sigma}(G)\right)^2
    \right] &=
    \E\left[
        \left(\sum_a g_a \cN^{*,\sigma}(G_a)\right)^2
    \right]  \\
    &= \sum_a  \left( \cN^{*,\sigma}(G_a) \right)^2 .
\end{align}
\end{proof}

\begin{lemma}\label{lemma:weighted-Q-replacer}
Let $\cN:\B{A}\to\B{B}$ be a quantum channel, and let $\sigma\in\cD_+(B)$ be a full rank state.
Then $Q_{\cN,\sigma}\geq \iden_A.$ Moreover, $Q_{\cN,\sigma}=\iden_A$ if and only if $\cN(X)=\Tr(X)\sigma$.
\end{lemma}

\begin{proof}
Note that $\norm{\sigma}_\sigma=1$ and $\cN^{*,\sigma}(\sigma)
    =
    \cN^*(\sigma^{-1/2}\sigma\sigma^{-1/2})
    =
    \cN^*(\iden_B)
    =
    \iden_A.$ Hence, we can write a standard Gaussian vector in $(\Bsa{B},\langle\cdot,\cdot\rangle_\sigma)$ as
\begin{equation}
    G_\sigma=g_0\sigma+G_\sigma^\perp,
\end{equation}
where $g_0\sim N(0,1)$ and $G_\sigma^\perp \perp_{\sigma} \sigma$ is independent of $g_0$.
Then,
\begin{align}
    Q_{\cN,\sigma}
    &=
    \E\left[
        \left(
            g_0\iden_A+\cN^{*,\sigma}(G_\sigma^\perp)
        \right)^2
    \right] \\
    &=
    \iden_A+
    \E\left[
        \left(\cN^{*,\sigma}(G_\sigma^\perp)\right)^2
    \right]
    \geq
    \iden_A.
\end{align}
Equality holds if and only if $\cN^{*,\sigma}$ annihilates
\begin{equation}
    \sigma^\perp
    :=
    \{Y\in\Bsa{B}: \langle Y,\sigma\rangle_\sigma=0\}
    =
    \{Y\in\Bsa{B}: \Tr Y=0\},
\end{equation}
which is equivalent to $\cN(X)$ being $\sigma$-orthogonal to every traceless $Y$ for every
$X$, meaning $\cN(X)\in \operatorname{span}\{\sigma\}$. Since $\cN$ is trace-preserving,
this is equivalent to $\cN(X)=\Tr(X)\sigma.$
\end{proof}

Next, we note a simple multiplicativity property of the singular operator.

\begin{lemma}\label{lemma:weighted-Q-multiplicativity}
Let $\cN:\B{A}\to\B{B}$ and $\cM:\B{C}\to\B{D}$ be quantum channels, and let
$\sigma\in\cD_+(B), \omega\in\cD_+(D)$ be full rank states. Then
\begin{equation}
    Q_{\cN\otimes\cM,\sigma\otimes\omega}
    =
    Q_{\cN,\sigma}\otimes Q_{\cM,\omega}.
\end{equation}
\end{lemma}
\begin{proof}
Let $\{G_a\}_a\subseteq \Bsa{B}$ and $\{H_b\}_b \subseteq \Bsa{D}$ be $\sigma$-orthonormal and $\omega$-orthonormal bases of $\Bsa{B}$ and $\Bsa{D}$, respectively. Then, $\{ G_a \otimes H_b\}_{a,b} \subseteq \Bsa{B\otimes D}$ is a $\sigma\otimes \omega$-orthonormal basis of $\Bsa{B\otimes D}$. Hence, using Lemma~\ref{lemma:weighted-Q-SVD}, we get
\begin{align}
        Q_{\cN\otimes \cM,\sigma\otimes \omega}
    &=
    \sum_{a,b} \bigl((\cN\otimes \cM)^{*, \sigma \otimes \omega}(G_a \otimes H_b) \bigr)^2 \\ 
    &=  \sum_{a,b}
        \bigl(\cN^{*,\sigma}(G_a) \otimes \cM^{*,\omega}  (H_b) \bigr)^2
     \\
    &= \sum_{a}
        \bigl(\cN^{*,\sigma}(G_a) \bigr)^2  \otimes \sum_b \bigl(\cM^{*,\omega}  (H_b) \bigr)^2
    = Q_{\cN, \sigma} \otimes Q_{\cM,\omega}.
\end{align}    
\end{proof}

Finally, we prove a crucial convexity property of the map $\sigma\mapsto Q_{\cN,\sigma}$.

\begin{lemma}\label{lemma:weighted-Q-convex}
Fix a quantum channel $\cN:\B{A}\to\B{B}$. Then, the map 
\begin{equation}
    \cD_+(B) \ni \sigma\mapsto Q_{\cN,\sigma}\in \Bsa{A}
\end{equation}
is convex in the positive semi-definite order: $Q_{\cN,p\sigma_1+(1-p)\sigma_2}
    \leq
    p Q_{\cN,\sigma_1}
    +(1-p)Q_{\cN,\sigma_2}.$
\end{lemma}

\begin{proof}
Recall that $W_\sigma(Y):=\sigma^{-1/2}Y\sigma^{-1/2}=\Phi_{\sigma}^2(Y)$ for $ Y\in\Bsa B$, so that
\begin{equation}
    \langle Y,Z\rangle_\sigma
    = \langle \Phi_{\sigma}(Y),\Phi_{\sigma}(Z)\rangle_2
    =\langle Y,W_\sigma(Z)\rangle_2 = \Tr (Y W_{\sigma}(Z)).
\end{equation}
If \(G_\sigma\) is a standard Gaussian vector in
$(\Bsa B,\langle\cdot,\cdot\rangle_\sigma)$, then \(Z_\sigma:=W_\sigma(G_\sigma)\)
is a centered Gaussian vector in the Hilbert-Schmidt space $(\Bsa B,\langle\cdot,\cdot\rangle)$ with
covariance \(W_\sigma\). Moreover,
\begin{align}
Q_{\cN,\sigma} = \E \left[ \big( \cN^{*,\sigma}(G_\sigma) \big)^2 \right]
    =
    \E \left[\big(\cN^*(Z_\sigma)\big)^2 \right] &=
    \sum_{a,b}
    W_{ab}\,
    \cN^*(E_a)\cN^*(E_b) \label{eq:weighted-singular-Wab} \\
    &=: \Xi_{\cN}(W_{\sigma}),
\end{align}
where \(\{E_b\}_b\) is any Hilbert-Schmidt orthonormal basis of \(\Bsa B\) and $W_{ab}
    :=
    \langle E_a,W_\sigma(E_b)\rangle.$
Thus, \(Q_{\cN,\sigma}\) is obtained from \(W_\sigma\) by a fixed positivity-preserving linear map $\Xi_{\cN}$. Hence, it suffices to prove that $\sigma\mapsto W_\sigma$
is convex.

In a fixed basis, vectorization shows $W_\sigma
    \simeq
    \sigma^{-1/2}\otimes \overline{\sigma}^{-1/2}.$ By Ando's convexity theorem \cite{Ando1979convex}, the map $(A,B)\mapsto A^{-1/2}\otimes B^{-1/2}$
is jointly convex. Applying this
to our case immediately yields the desired claim.
\end{proof}

\subsection{Gaussian mean width behaviour}

With the preliminary material in place, we are now ready to state one of our main results, which describes the asymptotic behavior of the weighted Gaussian widths $w_{G,\sigma^{\otimes n}}(I_n(\cN))$ of the channel images
\begin{equation}\label{eq:N-image}
    I_n(\cN) := \cN^{\otimes n} \left( \cD(A^{\otimes n}) \right) \subseteq \Bsa{B^{\otimes n}},
\end{equation}
where $\cN : \B{A}\to \B{B}$ is a quantum channel, $\sigma\in \cD_+(B)$ is a full rank state, and 
\begin{equation}
    w_{G,\sigma^{\otimes n}}(I_n(\cN)) := \E \sup_{\rho\in \cD(A^{\otimes n})} \left\langle G_n, \cN^{\otimes n}(\rho) \right\rangle_{\sigma^{\otimes n}},
\end{equation}
with $G_n$ being a standard Gaussian vector in $(\Bsa{B^{\otimes n}},\langle\cdot,\cdot\rangle_{\sigma^{\otimes n}})$. We will assume, without loss of generality, that the channel is `minimally' defined \cite[Remark 2.2]{Singh2022detecting}, in the sense that for any full rank input state $\eta\in \cD_+ (A)$, the output $\cN(\eta)\in \cD_+(B)$ has full rank. 

\begin{theorem}[Weighted Gaussian mean width asymptotic behavior] 
\label{theorem:weighted-gaussian-width-growth} \hspace{2pt} \\ 
Let $\cN:\B{A}\to\B{B}$ be a quantum channel with image $I_n(\cN)$ as in \eqref{eq:N-image}. Let $\eta\in\cD_+(A)$ be a full rank state, $\sigma:=\cN(\eta)$, and $Q_{\cN,\sigma}$ be the singular operator from \eqref{eq:weighted-singular-Q}. Then, 
\begin{equation}
 \forall n\in \N: \quad   \frac{1}{\sqrt{2\pi}}
    \sqrt{\norm{Q_{\cN,\sigma}}_\infty^n-1}
    \leq
    w_{G,\sigma^{\otimes n}}\bigl(I_n(\cN)\bigr)
    \leq
    \sqrt{2n\log d_A
    \left(\norm{Q_{\cN,\sigma}}_\infty^n-1 \right)}.
\end{equation}
\end{theorem}

\begin{proof}
Let $G_n =
    g_0\sigma^{\otimes n}+G_n^\perp$ be a standard random Gaussian vector in the weighted Euclidean space $\bigl(\Bsa{B^{\otimes n}}, \langle\cdot,\cdot\rangle_{\sigma^{\otimes n}}\bigr)$, where $g_0\sim N(0,1)$ and
$G_n^\perp\perp_{\sigma^{\otimes n}}\sigma^{\otimes n}$ is independent of $g_0$. Since Gaussian width is translation-invariant, we can write
\begin{align}
    w_{G,\sigma^{\otimes n}}\bigl(I_n(\cN)\bigr)
    &=
    w_{G,\sigma^{\otimes n}}
    \left(
        I_n(\cN)-\sigma^{\otimes n}
    \right) \\
    &=
    \E \sup_{\rho\in\cD(A^{\otimes n})}
    \left\langle
        G_n,
        \cN^{\otimes n}(\rho-\eta^{\otimes n})
    \right\rangle_{\sigma^{\otimes n}} \\
    &=
    \E \sup_{\rho\in\cD(A^{\otimes n})}
    \Tr\left(
        X_n(\rho-\eta^{\otimes n})
    \right) \\
    &= \E\lambda_{\max}\left(
        X_n-\Tr(\eta^{\otimes n}X_n)\iden_A^{\otimes n}
    \right) \\
    &= \E \lambda_{\max}(X_n^{\circ}),
\end{align}
where we introduced
\begin{align}
    X_n &:=
    (\cN^{\otimes n})^{*,\sigma^{\otimes n}}(G_n)=g_0\iden_A^{\otimes n}
    +
    (\cN^{\otimes n})^{*,\sigma^{\otimes n}}(G_n^\perp) , \\ 
    X_n^\circ &:=
    X_n-\Tr(\eta^{\otimes n}X_n)\iden_A^{\otimes n} =
    (\cN^{\otimes n})^{*,\sigma^{\otimes n}}(G_n^\perp).
\end{align}
Now, by fixing a $\sigma^{\otimes n}$-orthonormal basis $\{G_\beta\}_\beta$ of $\left(\sigma^{\otimes n}\right)^\perp :=
    \left\{
        Y\in\Bsa{B^{\otimes n}}:
        \left\langle
            Y,\sigma^{\otimes n}
        \right\rangle_{\sigma^{\otimes n}}=0
    \right\}$, we can decompose $
    G_n^\perp=\sum_\beta g_\beta G_\beta,$
where $(g_\beta)_\beta$ are i.i.d standard real Gaussian variables. Hence, by using the matrix Gaussian-series estimate
(Lemma~\ref{lemma:gaussian-series-bound}), we obtain
\begin{align}
    w_{G,\sigma^{\otimes n}}\bigl(I_n(\cN)\bigr)
    &=
    \E\lambda_{\max}(X_n^\circ) \\
    &\leq
    \sqrt{2n\log d_A} 
    \norm{
        \sum_\beta
        \left[
            (\cN^{\otimes n})^{*,\sigma^{\otimes n}}(G_\beta)
        \right]^2
    }_\infty^{1/2} \\
    &= \sqrt{2n\log d_A \norm{\E (X_n^\circ)^2}_{\infty}} \\
    &= \sqrt{2n\log d_A  \left(\norm{Q_{\cN,\sigma}}_\infty^n-1 \right)}, \label{eq:weighted-wG-upper}
\end{align}
where we use $\E (X_n^{\circ})^2 = \E X_n^2- \iden_A^{\otimes n}$ and $\E X_n^2=Q_{\cN^{\otimes n}, \sigma^{\otimes n}}=Q_{\cN,\sigma}^{\otimes n} \geq \iden_A^{\otimes n}$ (Lemmas~\ref{lemma:weighted-Q-SVD}-\ref{lemma:weighted-Q-multiplicativity}).

For the lower bound, let $\ket{\psi_n}$ be a unit eigenvector of
$Q_{\cN,\sigma}^{\otimes n}-\iden_A^{\otimes n}$ corresponding to its largest eigenvalue
$\norm{Q_{\cN,\sigma}}_\infty^n-1$. Define $\ket{\xi_n}:=X_n^\circ\ket{\psi_n}.$
Then,
\begin{align}
    \E\norm{\xi_n}^2
    =
    \E\bra{\psi_n}(X_n^\circ)^2\ket{\psi_n} =
    \bra{\psi_n}
    \left(
        Q_{\cN,\sigma}^{\otimes n}
        -
        \iden_A^{\otimes n}
    \right)
    \ket{\psi_n} =
    \norm{Q_{\cN,\sigma}}_\infty^n-1.
\end{align}
Using the Gaussian Kahane-Khintchine inequality for Hilbert spaces (Lemma~\ref{lemma:gaussian-L1-L2-complex})
\begin{align}
    \E\norm{X_n^\circ}_\infty
    \geq
    \E\norm{X_n^\circ\ket{\psi_n}} 
    =
    \E\norm{\xi_n} 
    \geq
    \sqrt{\frac{2}{\pi} \E\norm{\xi_n}^2} 
    =
    \sqrt{\frac{2}{\pi} \left(\norm{Q_{\cN,\sigma}}_\infty^n-1\right)}.
\end{align}
On the other hand, $\Tr(\eta^{\otimes n}X_n^\circ)=0$, and
since $\eta$ is full rank, this implies that $\lambda_{\max}(X_n^\circ)\geq 0$ and $
    \lambda_{\max}(-X_n^\circ)\geq 0.$ Hence, $\norm{X_n^\circ}_\infty
    =
    \max\bigl\{
        \lambda_{\max}(X_n^\circ),
        \lambda_{\max}(-X_n^\circ)
    \bigr\}
    \leq
    \lambda_{\max}(X_n^\circ)+\lambda_{\max}(-X_n^\circ).$
Moreover, since $X_n^\circ$ and $-X_n^\circ$ have the same distribution, we get
\begin{equation}
    \E\norm{X_n^\circ}_\infty
    \leq
    2\,\E\lambda_{\max}(X_n^\circ)
    =
    2\,w_{G,\sigma^{\otimes n}}\bigl(I_n(\cN)\bigr).
\end{equation}
Combining the last two estimates gives
\begin{equation}\label{eq:weighted-wG-lower}
    w_{G,\sigma^{\otimes n}}\bigl(I_n(\cN)\bigr)
    \geq
    \frac{1}{\sqrt{2\pi}}\,
    \sqrt{\norm{Q_{\cN,\sigma}}_\infty^n-1}.
\end{equation}
The bounds \eqref{eq:weighted-wG-upper} and \eqref{eq:weighted-wG-lower} together prove the asserted
two-sided estimate.
\end{proof}

\begin{remark}\label{remark:weighted-width-upper}
For a channel $\cN:\B{A}\to\B{B}$ and \emph{any} full rank state $\sigma\in\cD_+(B)$, the argument used in Theorem~\ref{theorem:weighted-gaussian-width-growth} gives the following upper bound on the weighted widths:
\begin{equation}\label{eq:weighted-width-upper-general}
\forall n\in \N: \quad    w_{G,\sigma^{\otimes n}}\bigl(I_n(\cN)\bigr)
    \leq
    \sqrt{2n\log d_A}\,
    \norm{Q_{\cN,\sigma}}_\infty^{n/2}.
\end{equation}
For the corresponding lower bound, we need $\sigma=\cN(\eta)$ for some full rank input $\eta\in \cD_+(A)$. 
\end{remark}

Combining Theorem~\ref{theorem:weighted-gaussian-width-growth} with Sudakov and Dudley inequalities (Lemmas~\ref{lemma:sudakov} and \ref{lemma:dudley}) provides non-trivial bounds on the fixed scale weighted covering numbers of $I_n(\cN)$.

\begin{corollary}\label{corollary:weighted-covering-width-growth}
Let $\cN:\B{A}\to\B{B}$ be a quantum channel, let $\eta\in\cD_+(A)$ be a full rank state, and $\sigma:=\cN(\eta)$. Then, there
exist constants $\mathscr{C}, \mathscr{C}'>0$ depending only on the dimensions $d_A, d_B$, such that
for every $\eps>0$ and $n\in\N$,
\begin{equation}
    \log C_{\eps}
    \bigl(
        I_n(\cN);
        \norm{\cdot}_{\sigma^{\otimes n}}
    \bigr)
    \leq
    \frac{\mathscr{C} n}{\eps^2}
    \left(
        \norm{Q_{\cN,\sigma}}_\infty^n-1
    \right).
\end{equation}
Moreover,
\begin{equation}
    \sup_{\eps>0}
    \eps^2
    \log C_{\eps}
    \bigl(
        I_n(\cN);
        \norm{\cdot}_{\sigma^{\otimes n}}
    \bigr)
    \geq
    \frac{\mathscr{C}'}{n^2}
    \left(
        \norm{Q_{\cN,\sigma}}_\infty^n-1
    \right).
\end{equation}
\end{corollary}

\begin{proof}
Using Sudakov's inequality (Lemma~\ref{lemma:sudakov}) in the Euclidean space $\bigl(
        \Bsa{B^{\otimes n}},
        \langle\cdot,\cdot\rangle_{\sigma^{\otimes n}}
    \bigr),$
together with Theorem~\ref{theorem:weighted-gaussian-width-growth}, we get
\begin{align}
    \log C_{\eps}
    \bigl(
        I_n(\cN);
        \norm{\cdot}_{\sigma^{\otimes n}}
    \bigr)
    &\leq
    C\,
    \frac{
        w_{G,\sigma^{\otimes n}}(I_n(\cN))^2
    }{\eps^2} \\
    &\leq
    \frac{2nC\log d_A}{\eps^2}
    \left(
        \norm{Q_{\cN,\sigma}}_\infty^n-1
    \right).
\end{align}
This proves the first claim with $\mathscr{C}=2C\log d_A$. Similarly, the second claim follows from Dudley's inequality (Lemma~\ref{lemma:dudley}), together with
Eq.~\eqref{eq:sudakov-dudley} and
Theorem~\ref{theorem:weighted-gaussian-width-growth}.
\end{proof}

\subsection{Identification strong converse}

We now apply our main result from the previous section to obtain a single-letter strong converse bound on the identification capacity of quantum channels (c.f. Definition~\ref{def:IDcodes}).

\begin{theorem}[Weighted Gaussian converse bound on identification] 
\label{theorem:weighted-gaussian-converse} \hspace{2pt} \\
Let $\cN:\B{A}\to\B{B}$ be a quantum channel, and $\sigma\in\cD_+(B)$ be a full rank state. Let $I_n(\cN)$ be the channel image from \eqref{eq:N-image} and $Q_{\cN,\sigma}$ be the singular operator from \eqref{eq:weighted-singular-Q}.
Then, for every $\lambda_1,\lambda_2>0$ with $\lambda_1+\lambda_2<1$,
\begin{align}
    \limsup_{n\to\infty}
    \frac{1}{n}
    \log\log N_{(n,\lambda_1,\lambda_2)}(\cN)
    &\leq
    \max\left\{
        0,\;
        2\limsup_{n\to\infty}
        \frac{1}{n}
        \log
        w_{G,\sigma^{\otimes n}}\bigl(I_n(\cN)\bigr)
    \right\} \\
    &\leq \log\norm{Q_{\cN,\sigma}}_\infty.
\end{align}
Consequently,
\begin{equation}\label{eq:inf_sigmaQ_N}
    \limsup_{n\to\infty}
    \frac{1}{n}
    \log\log N_{(n,\lambda_1,\lambda_2)}(\cN)
    \leq
    \inf_{\sigma\in\cD_+(B)} \biggl(
    \log\norm{Q_{\cN,\sigma}}_\infty \biggr).
\end{equation}
\end{theorem}

\begin{proof}
Fix $0<\zeta<1-\lambda_1-\lambda_2$, and consider a $(n,N,\lambda_1,\lambda_2)$ identification code $(\rho_i, D_i)_{i\in [N]}$ for $\cN$ (c.f. Definition~\ref{def:IDcodes}). If two code outputs
$\omega_i=\cN^{\otimes n}(\rho_i),\omega_j=\cN^{\otimes n}(\rho_j)\in I_n(\cN)$ lie in the same
$\norm{\cdot}_{\sigma^{\otimes n}}$-ball of radius $\zeta$, then
\begin{equation}
    \frac{1}{2}\norm{\omega_i-\omega_j}_1
    \leq
    \frac{1}{2}\norm{\omega_i-\omega_j}_{\sigma^{\otimes n}}
    \leq
    \zeta
    <
    1-\lambda_1-\lambda_2,
\end{equation}
where we used $\norm{Y}_1\leq\norm{Y}_{\sigma^{\otimes n}}$ (Lemma~\ref{lemma:weighted-norm-dominates-trace}). This contradicts
Eq.~\eqref{eq:CID-separation} and Lemma~\ref{lemma:CID-separation}. Hence $N_{(n,\lambda_1,\lambda_2)}(\cN)
    \leq
    C_{\zeta}
    \bigl(
        I_n(\cN);
        \norm{\placeholder}_{\sigma^{\otimes n}}
    \bigr).$ By Sudakov's inequality (Lemma~\ref{lemma:sudakov}) in the Euclidean space $\bigl(\Bsa{B^{\otimes n}},
    \langle\cdot,\cdot\rangle_{\sigma^{\otimes n}}\bigr),$
there exists a universal constant $C>0$ such that
\begin{align}
    \log N_{(n,\lambda_1,\lambda_2)}(\cN)
    &\leq
    \log C_{\zeta}
    \bigl(
        I_n(\cN);
        \norm{\placeholder}_{\sigma^{\otimes n}}
    \bigr) \\
    &\leq
    C\,
    \frac{
        w_{G,\sigma^{\otimes n}}(I_n(\cN))^2
    }{\zeta^2}.
\end{align}
Taking logarithms, dividing by $n$, and letting $n\to \infty$ shows
\begin{align}
    \limsup_{n\to\infty}
    \frac{1}{n}
    \log\log N_{(n,\lambda_1,\lambda_2)}(\cN)
    &\leq
    \max\left\{
        0,\;
        2\limsup_{n\to\infty}
        \frac{1}{n}
        \log
        w_{G,\sigma^{\otimes n}}\bigl(I_n(\cN)\bigr)
    \right\} \\
    &\leq \log\norm{Q_{\cN,\sigma}}_\infty,
\end{align}
where the second bound follows directly from Theorem~\ref{theorem:weighted-gaussian-width-growth} and Remark~\ref{remark:weighted-width-upper}. Since $\sigma\in \cD_+(B)$ was arbitrary, the final bound follows by optimizing over $\sigma$.
\end{proof}

Note that Lemma~\ref{lemma:weighted-Q-convex} shows that the optimization in Eq.~\eqref{eq:inf_sigmaQ_N} is a convex program.
We now show that it also admits a semidefinite representation. Fix $\sigma\in \cD_+(B)$ and define linear maps $L_{\sigma}, R_{\sigma}:\B{B}\to \B{B}$ as $L_\sigma(X):=\sigma X$ and $R_\sigma(X):=X\sigma$. These are (Hilbert-Schmidt) positive definite linear operators on $\B{B}$ that commute with each other. Hence, $L_\sigma\# R_\sigma = L_\sigma^{1/2} R_\sigma^{1/2}$, where $\#$ denotes the operator geometric mean \cite{Ando1979convex, Kubo1980mean}, and 
\begin{equation}
  W_\sigma(X)
    =
    \sigma^{-1/2}X\sigma^{-1/2}
    =
    (L_\sigma\# R_\sigma)^{-1}(X).
\end{equation}
Now, \cite[Theorem 4.1.3]{bhatia2015positive} gives the following extremal characterization of $L_{\sigma}\# R_{\sigma}$:
\begin{equation}
    L_{\sigma}\#R_{\sigma} = \max \left\{ H : H=H^*, \begin{pmatrix}
        L_\sigma & H\\
        H & R_\sigma
    \end{pmatrix} \geq 0 \right\}.
\end{equation}
Moreover, by taking Schur complements \cite[Theorem 1.3.3]{bhatia2015positive}, we know
\begin{equation}
    W\geq H^{-1}
    \quad\Longleftrightarrow\quad
    \begin{pmatrix}
        W & \id_B\\
        \id_B & H
    \end{pmatrix}
    \geq0 ,
\end{equation}
where $\id_B:\B{B}\to \B{B}$ is the identity map. Thus, $\inf_{\sigma\in\cD_+(B)}
    \norm{Q_{\cN,\sigma}}_\infty$ takes the form
\begin{align} \label{eq:converse-SDP}
    \textnormal{minimize}\quad & \theta \\
    \textnormal{subject to}\quad
    &\Xi_{\cN}(W)\leq \theta\iden_A, \nonumber\\
    &\begin{pmatrix}
        W & \id_B \\
        \id_B & H
    \end{pmatrix}\geq0, \nonumber\\
    &\begin{pmatrix}
        L_\sigma & H\\
        H & R_\sigma
    \end{pmatrix}\,\,\,\geq0, \nonumber\\
    &\,\,\sigma\geq0,\quad \Tr\sigma=1, \nonumber
\end{align}
where $W,H:\B{B}\to \B{B}$ are (Hilbert-Schmidt) self-adjoint linear maps, and
\begin{equation}
    \Xi_{\cN}(W)
    :=
    \sum_{a,b}
    \operatorname{Re}\Tr\left(E_a W(E_b)\right)
    \cN^*(E_a)\cN^*(E_b) \in \Bsa{A},
\end{equation}
for any Hilbert-Schmidt orthonormal basis $\{ E_b\}_b$ of \(\Bsa B\). The map $W\mapsto \Xi_{\cN}(W)$ is linear, and $L_\sigma, R_\sigma$ also depend linearly
on $\sigma$. Hence, all constraints are linear matrix inequalities. Moreover, feasibility
implies that
\begin{equation}
    0<H\leq L_\sigma\#R_\sigma, 
    \quad\text{and} \quad
    W\geq H^{-1} \geq (L_\sigma\#R_\sigma)^{-1}=W_\sigma.
\end{equation}
Since \(W\mapsto \Xi_{\cN}(W)\) is positivity-preserving and linear (see Eq.~\eqref{eq:weighted-singular-Wab}),
\begin{equation}
    \Xi_{\cN}(W)\geq \Xi_{\cN}(W_{\sigma}) = Q_{\cN,\sigma}.
\end{equation}
Therefore, if $\sigma\in \cD_+(B)$ is feasible, the SDP minimum must be at least as large as $\norm{Q_{\cN,\sigma}}_{\infty}$. Conversely, for any $\sigma\in \cD_+(B)$, choosing $H=L_\sigma\#R_\sigma$ and $W=(L_\sigma\#R_\sigma)^{-1}=W_\sigma$ is feasible and gives the original objective $\norm{Q_{\cN,\sigma}}_{\infty}$, since $\Xi_{\cN}(W_{\sigma})=Q_{\cN,\sigma}$ (see Eq.~\eqref{eq:weighted-singular-Wab}). 

\section{Examples}\label{sec:application}

In this section, we apply Theorem~\ref{theorem:weighted-gaussian-converse} to obtain strong converse bounds on the identification capacities of some important noisy channels and examine how they compare with existing converse bounds from the literature, which we first recall below.

\begin{proposition}[\cite{Atif2024CIDstrongconverse, ye2026strongconverseboundsclassical} Quantum/Classical capacity converse bounds on identification]  \label{prop:cap-converse}
    Let $\cN:\B{A}\to \B{B}$ be a quantum channel. For every $\lambda_1,\lambda_2>0$ such that $\lambda_1+\lambda_2 <1$,
\begin{align}
    \limsup_{n\to\infty}\frac{1}{n}\log\log N_{(n,\lambda_1,\lambda_2)}(\cN)
    &\leq \log d_A + C(\cN) \label{eq:CID<=logd+C} \\
      \limsup_{n\to\infty}\frac{1}{n}\log\log N_{(n,\lambda_1,\lambda_2)}(\cN)
    &\leq \log d_A + Q^{\dagger}(\cN) \label{eq:CID<=logd+Q},
\end{align}
where $C(\cN)$ and $Q^{\dagger}(\cN)$ are the \emph{classical} and \emph{strong converse quantum capacities} of $\cN$ \cite{Wilde2016}. 
\end{proposition}

\begin{proposition}[Weighted ellipsoid converse bound on identification]\label{prop:weighted-ellipsoid-converse}  \hspace{2pt} \\
Let $\cN:\B{A}\to\B{B}$ be a quantum channel and
$\sigma\in\cD_+(B)$ be full rank. Let $s_a(\cN,\sigma)$
denote the $\sigma$-weighted singular values of $\cN$.
Then, for every $\lambda_1,\lambda_2>0$ with $\lambda_1+\lambda_2<1$,
\begin{equation}
    \limsup_{n\to\infty}
    \frac1n\log\log N_{(n,\lambda_1,\lambda_2)}(\cN)
    \le
    \inf_{t>0}
    \log\left(\sum_a s_a(\cN,\sigma)^t\right).
\end{equation}
Consequently,
\begin{equation}
    \limsup_{n\to\infty}
    \frac1n\log\log N_{(n,\lambda_1,\lambda_2)}(\cN)
    \le
    \inf_{\sigma\in\cD_+(B)}
    \inf_{t>0}
    \log\left(\sum_a s_a(\cN,\sigma)^t\right).
\end{equation}
\end{proposition}

\begin{proof}
    This was proven for the qubit depolarizing channel for $\sigma=\iden/2$ recently in \cite{ye2026strongconverseboundsclassical}. The general result for arbitrary channels is proved in Appendix~\ref{appen:ellipsoid}, see Proposition~\ref{prop:weighted-ellipsoid-converse-appen}
\end{proof}

\subsection{Unital qubit channels}

Let $\mathcal N:\B{\C^2}\to \B{\C^2}$ be a unital qubit channel with Bloch
representation $r\mapsto Tr$, where $T\in\mathbb R^{3\times 3}$ \cite{BethRuskai2002qubit}. Let
$s_1,s_2,s_3$ be the singular values of $T$, so that the (Hilbert-Schmidt) singular values\footnote{Note that the $\iden/2$-weighted singular values satisfy $s_a(\cN, \iden/2)=\sqrt{2}s_a$ for each $a$.} of $\cN$ are given by $s_0=1,s_1, s_2, s_3$. Choose a singular value decomposition of $T$, i.e. orthonormal vectors
$v_a,u_a\in\R^3$ such that $T v_a=s_a u_a $ for $ a=1,2,3$. Then, the corresponding right and left singular operators of $\cN$ are $F_0=G_0=\iden /\sqrt{2}$, 
\begin{equation}
    F_a:=\frac{ (v_a)_x \sigma_x + (v_a)_y\sigma_y + (v_a)_z \sigma_z}{\sqrt2},
    \qquad
    G_a:=\frac{(u_a)_x \sigma_x + (u_a)_y\sigma_y + (u_a)_z \sigma_z}{\sqrt2},
\end{equation}
with $\mathcal N(F_a)=s_aG_a.$ Moreover, $F_a^2
    =
    \frac{\iden}{2}=G_a^2,$ so that the singular operator
\begin{equation}
    Q_{\mathcal N, \iden/2}
    =
    (1+s_1^2+s_2^2+s_3^2)\iden,
\end{equation}
see Lemma~\ref{lemma:weighted-Q-SVD}. It turns out that the maximally mixed state is optimal for the weighted Gaussian
bound in Theorem~\ref{theorem:weighted-gaussian-converse}. Indeed, for any full-rank state $\sigma\in\cD_+(\C^2)$,
\begin{align}
    \Tr Q_{\cN,\sigma}
    &=
    \Tr(\cN^{*,\sigma}\cN)
    =
    \sum_{a}\norm{\cN(F_a)}_\sigma^2  \\
    &=
    \sum_{a}s_a^2\norm{G_a}_\sigma^2
    \geq
    \sum_{a}s_a^2\norm{G_a}_1^2
    =
    2\sum_{a}s_a^2.
\end{align}
Therefore, $\norm{Q_{\cN,\sigma}}_\infty
    \geq
    \frac12\Tr Q_{\cN,\sigma}
    \geq
    \sum_{a}s_a^2 = \norm{Q_{\cN, \iden/2}}_{\infty}$, and Theorem~\ref{theorem:weighted-gaussian-converse} shows
\begin{align}
    \limsup_{n\to\infty}
    \frac1n\log\log N_{(n,\lambda_1,\lambda_2)}(\cN)
    &\leq \inf_{\sigma\in\cD_+(\C^2)} \biggl(
    \log\norm{Q_{\cN,\sigma}}_\infty \biggr) \\
    &=
    \log\left(
        1+s_1^2+s_2^2+s_3^2
    \right). \label{eq:weighted-Gaussian-unital-qubit}
\end{align}

We conjecture that the maximally mixed state is also optimal for the weighted ellipsoid bound in Proposition~\ref{prop:weighted-ellipsoid-converse}, in the sense that 
\begin{align}
    \inf_{\sigma\in\cD_+(\C^2)}
    \inf_{t>0}
    \log\left(
        \sum_{a=0}^3s_a(\cN,\sigma)^t
    \right)
    &=
    \inf_{t>0}
    \log\left(
        2^{t/2}\sum_{a=0}^3s_a^t
    \right) \\
    &=
    \inf_{t>0}
    \left[
        \frac{t}{2}
        +
        \log\bigl(1+s_1^t+s_2^t+s_3^t\bigr)
    \right].
\end{align}
In any case, Proposition~\ref{prop:weighted-ellipsoid-converse} with $\sigma=\iden/2$ shows
\begin{equation}\label{eq:weighted-Ellipsoid-unital-qubit}
    \limsup_{n\to\infty}
    \frac1n\log\log N_{(n,\lambda_1,\lambda_2)}(\cN)
    \leq
    \inf_{t>0}
    \left[
        \frac{t}{2}
        +
        \log\bigl(1+s_1^t+s_2^t+s_3^t\bigr)
    \right].
\end{equation}

Finally, let us note that the Gaussian bound is never worse than the ellipsoid bound. It suffices to show that for all $t>0$, $ 2^{t/2}\sum_{a}s_a^t
    \geq
    \sum_{a}s_a^2.$
For $0<t\leq2$, this follows from\footnote{Note that since the channel is unital, $0\leq s_a \leq 1$ for all $a$ \cite{Wolf2012Qtour}.} $s_a^t\geq s_a^2$. For $t\geq2$, monotonicity of vector
$\ell_p$-norms gives
\begin{equation}
    \sum_{a}s_a^2
    =
    \norm{s}_2^2
    \leq
    4^{1-2/t}\norm{s}_t^2
    =
    4^{1-2/t}
    \left(\sum_{a}s_a^t\right)^{2/t}.
\end{equation}
Since $4^{1-2/t}\leq 2^{t/2}$ and $\sum_a s_a^t\geq1$, we get $
    \sum_{a}s_a^2
    \leq
    2^{t/2}\sum_{a}s_a^t.$ Hence, 
\begin{equation}\label{eq:unital-qubit-Gaussian-vs-Ellipsoid}
   \log\left(
        1+s_1^2+s_2^2+s_3^2
    \right) \leq \inf_{t>0}
    \left[
        \frac{t}{2}
        +
        \log\bigl(1+s_1^t+s_2^t+s_3^t\bigr)
    \right] .
\end{equation}

We collect the above discussion in the following lemma.

\begin{lemma}
    Let $\cN:\B{\C^2}\to \B{\C^2}$ be a unital qubit channel with Hilbert-Schmidt singular values $s_a=s_a(\cN, \iden/2)/\sqrt{2} $. Then, for every $\lambda_1,\lambda_2>0$ with $\lambda_1+\lambda_2<1$,
    \begin{align}
      \limsup_{n\to\infty}
    \frac1n\log\log N_{(n,\lambda_1,\lambda_2)}(\cN) &\leq   \log\left(
        1+s_1^2+s_2^2+s_3^2
    \right) \\ 
    &\leq \inf_{t>0}
    \left[
        \frac{t}{2}
        +
        \log\bigl(1+s_1^t+s_2^t+s_3^t\bigr)
    \right] .
    \end{align}
\end{lemma}

\subsubsection{Depolarizing channel} \label{sec:depol}

Consider the qubit depolarizing channel $\cD_p:\B{\C^2}\to \B{\C^2}$ defined as 
\begin{align} \label{eq:qubit-depol}
    \cD_p(X) & =(1-p)X+p\Tr(X)\,\frac{\iden}{2},
    \qquad 0\leq p\leq 1 \\
    &= \left(1-\frac{3p}{4} \right)X + \frac{p}{4}\left( \sigma_x X\sigma_x + \sigma_y X\sigma_y + \sigma_z X\sigma_z\right).
\end{align}
In the Hilbert--Schmidt orthonormal Pauli basis
$\frac{\iden}{\sqrt{2}},
    \frac{\sigma_x}{\sqrt{2}},
    \frac{\sigma_y}{\sqrt{2}},
    \frac{\sigma_z}{\sqrt{2}},$
the channel acts diagonally:
\begin{equation}
    \cD_p\!\left(\frac{\iden}{\sqrt{2}}\right)=\frac{\iden}{\sqrt{2}},
    \qquad
    \cD_p\!\left(\frac{\sigma_a}{\sqrt{2}}\right)=(1-p)\frac{\sigma_a}{\sqrt{2}},
    \qquad a\in \{x,y,z\}.
\end{equation}
Hence, the (Hilbert-Schmidt) singular values of $\cD_p$ are $1, 1-p, 1-p, 1-p$. Hence, for every $\lambda_1,\lambda_2>0$ such that $\lambda_1+\lambda_2<1$, the Gaussian-width bound from Eq.~\eqref{eq:weighted-Gaussian-unital-qubit} shows
\begin{align} 
    \limsup_{n\to\infty}\frac{1}{n}\log\log N_{(n,\lambda_1,\lambda_2)}(\cD_p) \leq 
    \log\bigl(1+3(1-p)^2\bigr). \label{eq:qubit-depolarizing-gaussian}
\end{align}
Similarly, the Ellipsoid bound from Eq.~\eqref{eq:weighted-Ellipsoid-unital-qubit} shows
\begin{align}
    \limsup_{n\to\infty}\frac{1}{n}\log\log N_{(n,\lambda_1,\lambda_2)}(\cD_p)
    &\leq
    \inf_{t>0}\left[
        \frac{t}{2}
        +
        \log\bigl(1+3(1-p)^t\bigr)
    \right]. \\
    &=  \begin{cases}
        2, & 0\leq p\leq 1-2^{-2/3},\\[1mm]
        2-D\!\left(\gamma(p)\middle\| \frac{3}{4}\right), & 1-2^{-2/3}<p<1, 
    \end{cases} \label{eq:qubit-depolarizing-ellipsoid}
\end{align}
where the last equality follows from a variational optimization (Lemma~\ref{lemma:depolarizing-optimization} and \cite{ye2026strongconverseboundsclassical}).

Finally, the capacity bounds from Proposition~\ref{prop:cap-converse} show
\begin{align}
    \limsup_{n\to\infty}\frac{1}{n}\log\log N_{(n,\lambda_1,\lambda_2)}(\cD_p)
    &\leq 1 + C(\cD_p) = 2 - h\left(\frac{p}{2} \right), \label{eq:CID(Dp)<1+C(Dp)} \\
    \limsup_{n\to\infty}\frac{1}{n}\log\log N_{(n,\lambda_1,\lambda_2)}(\cD_p) &\leq 1+Q^{\dagger}(\cD_p) \\
     &\leq \begin{cases}
        2 - h\left(1-\frac{3p}{4} \right), \quad 0\leq p \leq 2/3 \\
        1, \qquad\qquad\qquad\quad 2/3\leq p \leq 1,
    \end{cases}  \label{eq:CID(Dp)<1+Q(Dp)}
\end{align}
where we use the fact that $C(\cD_p)=1-h(p/2)$ \cite{King2003depolarizing}, and $Q^{\dagger}(\cD_p)\leq R(\cD_p)$, where $R(\cdot)$ is the \emph{Rains} information\footnote{There is no explicit expression known for the strong-converse quantum capacity $Q^{\dagger}$ in general.}, which serves as a strong converse bound on the quantum capacity \cite{Tomamichel2017Qstrong-converse}, and is computed for the qubit depolarizing channel in \cite{Rains1999}:
\begin{equation}
   R(\cD_p)= \begin{cases}
        1 - h\left(1-\frac{3p}{4} \right), \quad 0\leq p \leq 2/3, \\
        0, \qquad\qquad\qquad\quad 2/3\leq p \leq 1.
    \end{cases} 
\end{equation}
Here, $h(x):= -x\log x - (1-x)\log (1-x)$ denotes the binary entropy function. 

\begin{figure}[ht]
    \centering
    \includegraphics[width=0.7\linewidth]{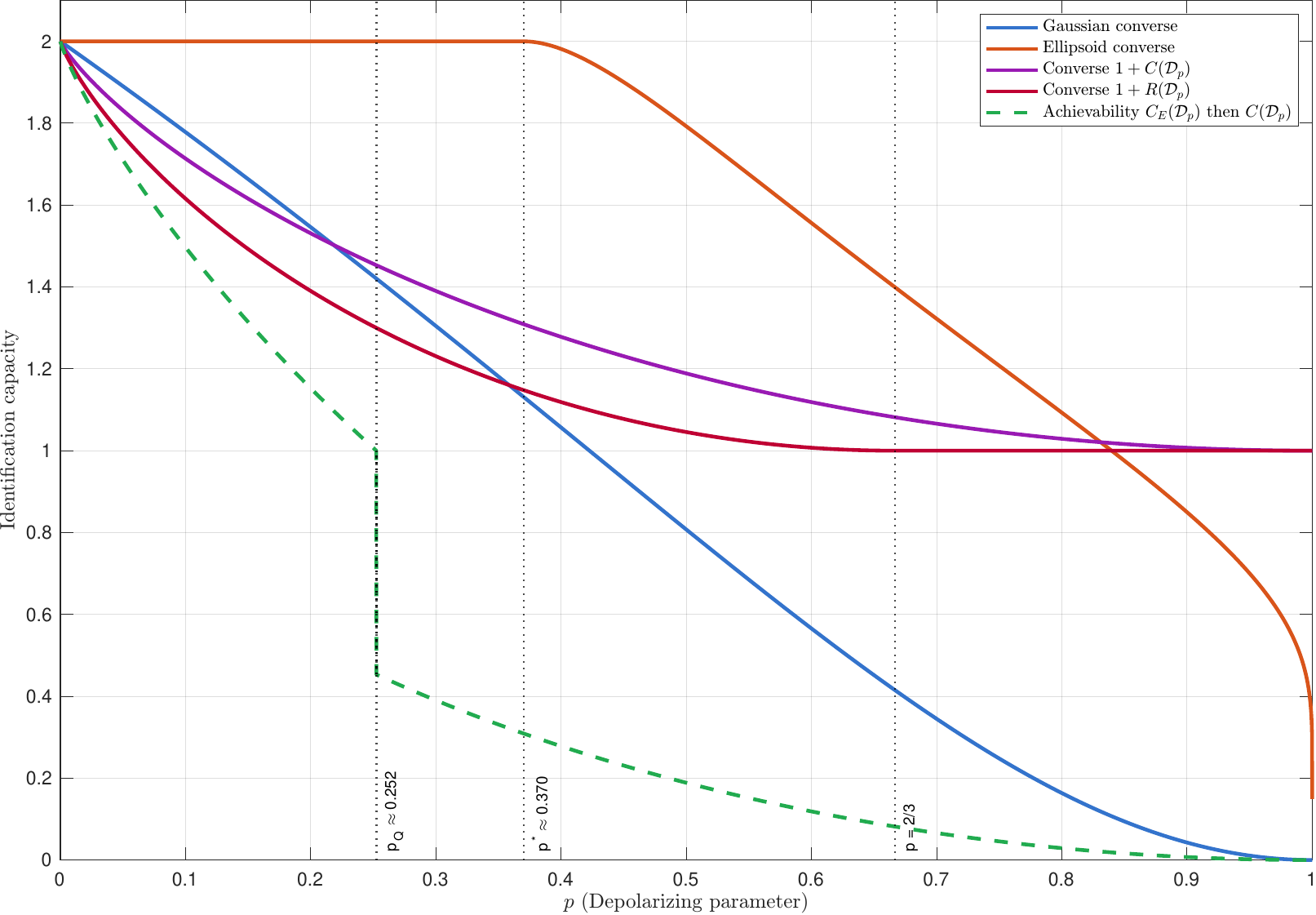}
    \caption{Strong converse and achievability bounds on the classical identification capacity of the qubit depolarizing channel $\cD_p$. The blue and orange curves show the Gaussian and Ellipsoid strong converse bounds from Eqs.~\eqref{eq:qubit-depolarizing-gaussian} and \eqref{eq:qubit-depolarizing-ellipsoid}, respectively. The purple and red curves show the classical capacity and quantum capacity strong converse bounds from Eqs.~\eqref{eq:CID(Dp)<1+C(Dp)} and \eqref{eq:CID(Dp)<1+Q(Dp)}, respectively. The dashed curve shows the achievability bound from Eq.~\eqref{eq:qubit-depolarizing-achievability}.}
    \label{fig:qubit-depolarizing}
\end{figure}

We compare these bounds with each other and with the best known achievability bound in Figure~\ref{fig:qubit-depolarizing}. For achievability, we combine the universal lower bound
$C_{\ID}(\cD_p)\geq C(\cD_p)$ with the stronger low-noise bound
$C_{\ID}(\cD_p)\geq C_E(\cD_p)$ from
\cite{Hayden2012QID-achievability} and
\cite[Remark 16 and Section 4]{Winter2013survey-ID}, where $C_E(\cdot)$ is the \emph{entanglement-assisted classical capacity} \cite{Wilde2016}. More precisely, applying
$\id_A\otimes(\cD_p)_{A\to B}$ to a maximally entangled state\footnote{Here, $A\cong B\cong \C^2$, and $\ket{\Omega}=\ket{00} + \ket{11}\in \C^2 \otimes \C^2$ is the unnormalized maximally entangled vector.}
$\frac12\ketbra{\Omega}_{AA}$ gives
\begin{align}
    \rho_{AB}
    &=
    \id_A\otimes(\cD_p)_{A\to B}
    \left(\frac12\ketbra{\Omega}_{AA}\right) \\
    &=
    \left(1-\frac{3p}{4}\right)
    \frac12\ketbra{\Omega}_{AB}
    +
    \frac p8 (\sigma_x)_B\ketbra{\Omega}_{AB}(\sigma_x)_B
    +
    \frac p8 (\sigma_y)_B\ketbra{\Omega}_{AB}(\sigma_y)_B \\
    & \hspace{7cm}+
    \frac{p}{8} (\sigma_z)_B\ketbra{\Omega}_{AB}(\sigma_z)_B , \nonumber
\end{align}
which has mutual and coherent information 
\begin{align}
    I(A:B)_\rho
    &=
    2-
    H\!\left(
        1-\frac{3p}{4},
        \frac p4,
        \frac p4,
        \frac p4
    \right)
    =
    C_E(\cD_p), \\
    I(A\rangle B)_\rho
    &=
    1-
    H\!\left(
        1-\frac{3p}{4},
        \frac p4,
        \frac p4,
        \frac p4
    \right),
\end{align}
respectively, where $H(\cdot)$ denotes the Shannon entropy of the input probability distribution. Thus, $I(A\rangle B)_{\rho}>0$ whenever $H\!\left(
        1-\frac{3p}{4},
        \frac p4,
        \frac p4,
        \frac p4
    \right)<1$, and \cite[Remark 16 + Section 4]{Winter2013survey-ID} shows
\begin{equation}
   C_{\ID}(\cD_p) \geq
    \begin{cases}
        C_E(\cD_p)=2-
        H\!\left(
            1-\frac{3p}{4},
            \frac p4,
            \frac p4,
            \frac p4
        \right),
        &
        H\!\left(
            1-\frac{3p}{4},
            \frac p4,
            \frac p4,
            \frac p4
        \right)<1, \\[3mm]
        C(\cD_p)=1-h\!\left(\dfrac p2\right),
        &
        H\!\left(
            1-\frac{3p}{4},
            \frac p4,
            \frac p4,
            \frac p4
        \right)\geq1 .
    \end{cases}
    \label{eq:qubit-depolarizing-achievability}
\end{equation}

\subsubsection{Pauli $XY$ channel}

Consider the qubit $XY$-Pauli channel $\cN_p:\B{\C^2}\to \B{\C^2}$ defined as
\begin{equation}
    \cN_p(X)
    =
    (1-p)X+\frac{p}{2}\sigma_x X\sigma_x+\frac{p}{2}\sigma_y X\sigma_y,
    \qquad 0\leq p\leq 1.
\end{equation}
In the Hilbert--Schmidt orthonormal Pauli basis $\frac{\iden}{\sqrt{2}},
    \frac{\sigma_x}{\sqrt{2}},
    \frac{\sigma_y}{\sqrt{2}},
    \frac{\sigma_z}{\sqrt{2}},$
the channel acts diagonally:
\begin{equation}
    \cN_p\!\left(\frac{\iden}{\sqrt{2}}\right)=\frac{\iden}{\sqrt{2}},
    \qquad
    \cN_p\!\left(\frac{\sigma_x}{\sqrt{2}}\right)
    =
    (1-p)\frac{\sigma_x}{\sqrt{2}},
\end{equation}
\begin{equation}
    \cN_p\!\left(\frac{\sigma_y}{\sqrt{2}}\right)
    =
    (1-p)\frac{\sigma_y}{\sqrt{2}},
    \qquad
    \cN_p\!\left(\frac{\sigma_z}{\sqrt{2}}\right)
    =
    (1-2p)\frac{\sigma_z}{\sqrt{2}}.
\end{equation}
Hence, the (Hilbert-Schmidt) singular values of $\cN_p$ are $1, 1-p, 1-p, |1-2p|.$ For every $\lambda_1,\lambda_2>0$ such that $\lambda_1+\lambda_2<1$, the Gaussian-width bound
from Eq.~\eqref{eq:weighted-Gaussian-unital-qubit} shows
\begin{align}
    \limsup_{n\to\infty}
    \frac{1}{n}\log\log N_{(n,\lambda_1,\lambda_2)}(\cN_p)
    &\leq
    \log\bigl(1+2(1-p)^2+(1-2p)^2\bigr) \nonumber\\
    &=
    \log\bigl(4-8p+6p^2\bigr).
    \label{eq:qubit-two-pauli-gaussian}
\end{align}

\begin{figure}[ht]
    \centering
    \includegraphics[width=0.7\linewidth]{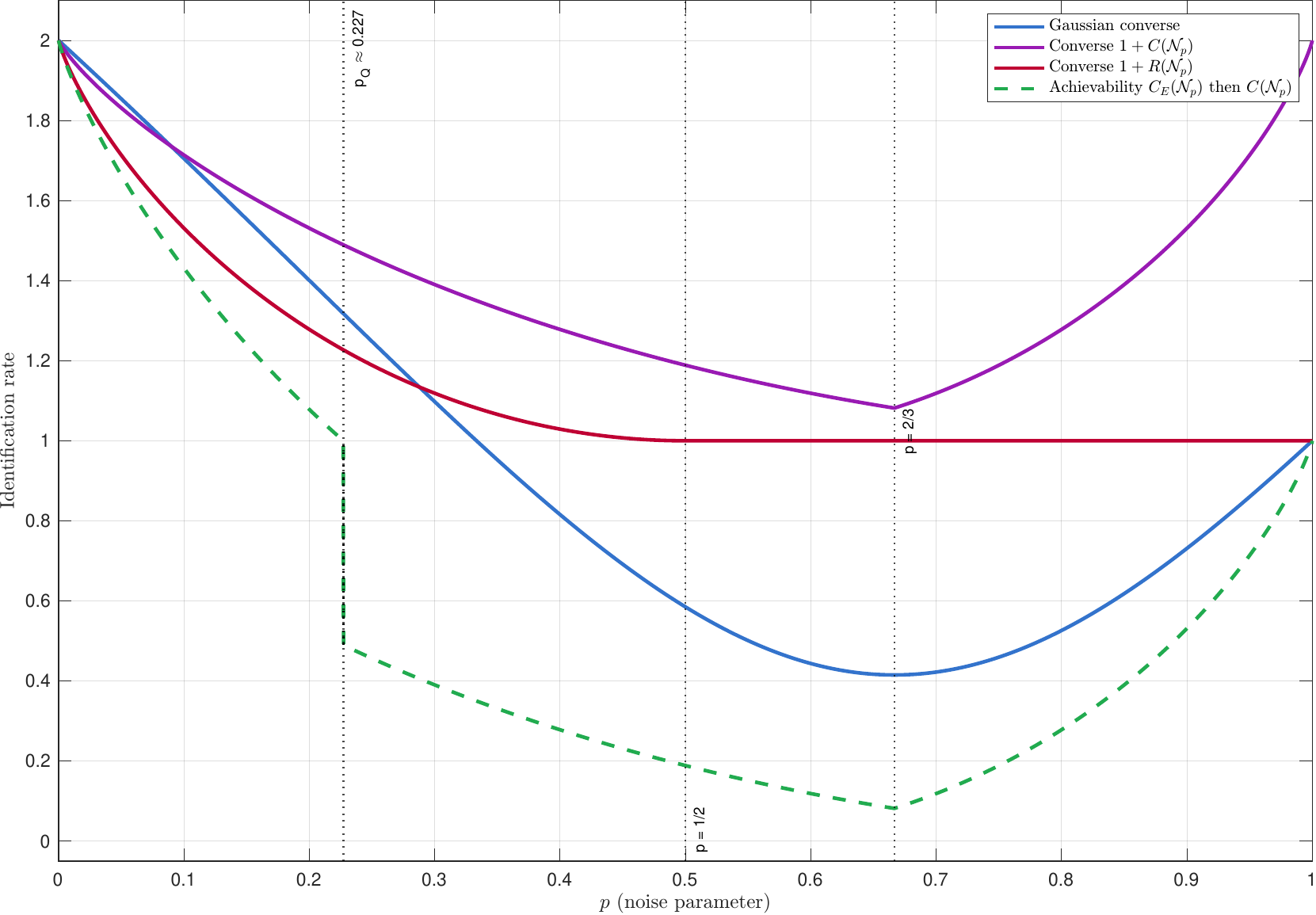}
    \caption{Strong converse and achievability bounds on the classical identification capacity of the qubit Pauli $XY$ channel $\cN_p$. The blue curve shows the Gaussian strong converse bound from Eq.~\eqref{eq:qubit-two-pauli-gaussian}. The purple and red curves show the capacity converse bounds from Eqs.~\eqref{eq:CID-two-pauli<=1+C} and \eqref{eq:CID-two-pauli<=1+Q}. The dashed curve shows the achievability bound from Eq.~\eqref{eq:qubit-two-pauli-achievability}.}
    \label{fig:qubit-two-pauli}
\end{figure}

We do not explicitly compute the Ellipsoid bound, since it cannot improve the Gaussian bound according to Eq.~\eqref{eq:unital-qubit-Gaussian-vs-Ellipsoid}. Finally, the capacity bounds (Proposition~\ref{prop:cap-converse}) show
\begin{align}
    \limsup_{n\to\infty}
    \frac{1}{n}\log\log N_{(n,\lambda_1,\lambda_2)}(\cN_p)
    &\leq
    1+C(\cN_p) \nonumber\\
    &=
    \begin{cases}
        2-h\!\left(\dfrac{p}{2}\right), & 0\leq p\leq \dfrac{2}{3},\\[2mm]
        2-h(p), & \dfrac{2}{3}\leq p\leq 1,
    \end{cases}
    \label{eq:CID-two-pauli<=1+C}
\end{align} 
and
\begin{align}
    \limsup_{n\to\infty}
    \frac{1}{n}\log\log N_{(n,\lambda_1,\lambda_2)}(\cN_p)
    &\leq
    1+Q^\dagger(\cN_p)
    \leq
    1+R(\cN_p) \nonumber\\
    &=
    \begin{cases}
        2-h(p), & 0\leq p\leq \dfrac{1}{2},\\[1mm]
        1, & \dfrac{1}{2}\leq p\leq 1.
    \end{cases}
    \label{eq:CID-two-pauli<=1+Q}
\end{align}
Here we used the fact that the classical capacity of a unital qubit channel is
\begin{equation}
    C(\cN)
    =
    1-h\!\left(\frac{1+s_{\max}}{2}\right),
\end{equation}
where $s_{\max}$ is the largest (Hilbert-Schmidt) singular value excluding 1 \cite{King2002additive, King2001unital-qubit-C}. For $\cN_p$, 
\begin{equation}
    s_{\max}(\cN_p)
    =
    \max\{1-p,|1-2p|\}
    =
    \begin{cases}
        1-p, & 0\leq p\leq \dfrac{2}{3},\\[1mm]
        2p-1, & \dfrac{2}{3}\leq p\leq 1.
    \end{cases}
\end{equation}
The Rains-information bound follows from the fact that the normalized Choi state of $\cN_p$ is Bell diagonal with eigenvalues $1-p, \frac{p}{2}, \frac{p}{2}, 0$:  
\begin{align}
    \rho_{AB} &= \id_A\otimes(\cN_p)_{A\to B} \left( \frac 12 \ketbra{\Omega}_{AA} \right) \\ 
    &= (1-p) \frac{1}{2}\ket{\Omega}\bra{\Omega}_{AB} + \frac{p}{4} (\sigma_x)_B\ket{\Omega}\bra{\Omega}_{AB} (\sigma_x)_B + \frac{p}{4} (\sigma_y)_B\ket{\Omega}\bra{\Omega}_{AB} (\sigma_y)_B.
\end{align}
Thus, \cite[Proposition 2]{Tomamichel2017Qstrong-converse}, \cite[Proposition 2]{Vedral1997E_R} and \cite{Miranowicz2008E_R-R} show that
\begin{equation}
    R(\cN_p)
    =
    \begin{cases}
        1-h(p), & 0\leq p\leq \dfrac{1}{2},\\[1mm]
        0, & \dfrac{1}{2}\leq p\leq 1.
    \end{cases}
\end{equation}

For achievability, note that $I(A:B)_\rho = 2-h(p)-p = C_E(\cN_p)$ and $
    I(A\rangle B)_\rho
    =
    1-h(p)-p.$
Thus, $I(A\rangle B)_\rho>0$ whenever $h(p)+p<1$, and \cite[Remark 16 + Section 4]{Winter2013survey-ID} shows
\begin{equation}
   C_{\ID}(\cN_p) \geq 
    \begin{cases}
        C_E(\cN_p)=2-h(p)-p, & h(p)+p<1, \\[1mm]
        C(\cN_p), & h(p)+ p \geq 1.
    \end{cases}
    \label{eq:qubit-two-pauli-achievability}
\end{equation}

We compare all the converse and achievability bounds with each other in Figure~\ref{fig:qubit-two-pauli}.

\subsection{Non-unital channels}

\subsubsection{Erasure channel}\label{sec:erasure}

Consider the qubit erasure channel $\cE_p:\B{\C^2}\to \B{\C^3}$
defined as
\begin{equation}
    \cE_p(X)
    =
    (1-p)X
    \oplus
    p\,\Tr(X)\ketbra{e},
    \qquad
    0\leq p\leq1,
\end{equation}
where $\ket{e}$ is an erasure flag orthogonal to the original qubit output space. The channel has the following covariance over the input $2\times 2$ unitary group $\mathbb{U}_2$ :
\begin{equation}\label{eq:erasure-covariance}
\forall U\in \mathbb{U}_2: \quad\cE_p(UXU^{\dagger}) = V_U \cE_p(X) V_U^{\dagger}, \quad V_U:= U \oplus \ketbra{e}.
\end{equation}
Now, fix a state $\sigma\in \cD_+(\C^3)$ and let \(\nu_U(\sigma)\equiv\nu:=V_U\sigma V_U^\dagger\). Then, the map
\(Y\mapsto V_UYV_U^\dagger\) is an isometry from
$(\mathcal L^{\rm sa}(\C^3),\langle\cdot,\cdot\rangle_\sigma)$ to
$(\mathcal L^{\rm sa}(\C^3),\langle\cdot,\cdot\rangle_\nu)$, since
\begin{align}
    \langle V_UYV_U^\dagger,V_UZV_U^\dagger\rangle_\nu
    &=
    \Tr\!\left(
        \nu^{-1/2}V_UYV_U^\dagger
        \nu^{-1/2}V_UZV_U^\dagger
    \right) \\  
    &=
    \Tr\!\left(
        \sigma^{-1/2}Y\sigma^{-1/2}Z
    \right)
    =
    \langle Y,Z\rangle_\sigma .
\end{align}
Thus, if \(G_\sigma\) is a standard Gaussian in $(\mathcal L^{\rm sa}(\C^3),\langle\cdot,\cdot\rangle_\sigma)$, then
\(G_\nu:=V_UG_\sigma V_U^\dagger\) is a standard Gaussian in $(\mathcal L^{\rm sa}(\C^3),\langle\cdot,\cdot\rangle_\nu)$. Moreover, the covariance \eqref{eq:erasure-covariance} implies 
\begin{equation}\label{eq:erasure-covariance-2}
\mathcal E_p^{*,\nu}(V_UYV_U^\dagger)
    =
    U\mathcal E_p^{*,\sigma}(Y)U^\dagger.    
\end{equation}
Hence, the singular operators (c.f. Eq.~\eqref{eq:weighted-singular-Q}) satisfy the relation
\begin{align}\label{eq:erasure-covariance-3}
    Q_{\mathcal E_p,\nu} \equiv Q_{\mathcal E_p,V_U\sigma V_U^{\dagger}}
    =
    \mathbb E\!\left[
        \mathcal E_p^{*,\nu}(G_\nu)^2
    \right]  =
    U\mathbb E\!\left[
        \mathcal E_p^{*,\sigma}(G_\sigma)^2
    \right]U^\dagger
    =
    UQ_{\mathcal E_p,\sigma}U^\dagger .
\end{align}
Then, using convexity of the map $\sigma \mapsto Q_{\cE_p, \sigma}$ (Lemma~\ref{lemma:weighted-Q-convex}), along with convexity and unitary invariance of the operator norm, we get 
\begin{equation}
    \norm{Q_{\cE_p, \tau(\sigma)}}_{\infty} \leq  \norm{Q_{\cE_p, \sigma}}_{\infty},
\end{equation}
where $\tau(\sigma):= \int_U V_U\sigma V_U^{\dagger} dU$ is the twirled state with respect to the normalized Haar measure on $\mathbb{U}_2$. Hence, in Theorem~\ref{theorem:weighted-gaussian-converse}, 
it suffices to optimize over reference states 
\begin{equation}
    \sigma_{q}
    :=
    q\,\frac{\iden}{2}
    \oplus
    (1-q)\ketbra{e},
    \qquad
    0<q<1.
\end{equation}

\begin{remark}\label{remark:covariance-twirl-convex}
    The argument given above applies to general group covariant channels. For example, for the depolarizing channel $\cD_p : \B{\C^d}\to \B{\C^d}$, it shows that the optimal Gaussian converse in Theorem~\ref{theorem:weighted-gaussian-converse} is achieved by the maximally mixed state (see also Lemma~\ref{lemma:mu-multiplicative-irreducible-covariance}).
\end{remark}

In the Hilbert-Schmidt space $(\Bsa{\C^2}, \langle \cdot , \cdot \rangle)$, fix the Pauli basis $F_0=\frac{\iden}{\sqrt{2}}, F_a = \frac{\sigma_a}{\sqrt{2}}$ for $a\in \{x,y,z \}$. In the weighted Euclidean space $(\Bsa{\C^3}, \langle \cdot , \cdot \rangle_{\sigma_q})$, we have
\begin{align}
    \norm{\cE_p(F_0)}_{\sigma_q}^2
    &=
    2\left(
        \frac{(1-p)^2}{q}
        +
        \frac{p^2}{1-q}
    \right), \\
    \norm{\cE_p(F_a)}_{\sigma_q}^2
    &=
    \frac{2(1-p)^2}{q},
    \qquad a\in\{x,y,z\}.
\end{align}
Thus, the $\sigma_q$-weighted singular values are given as
\begin{align}
    s_0(\cE_p,\sigma_q)^2
    &=
    2\left(
        \frac{(1-p)^2}{q}
        +
        \frac{p^2}{1-q}
    \right), \\
    s_a(\cE_p,\sigma_q)^2
    &=
    \frac{2(1-p)^2}{q},
    \qquad a\in\{x,y,z\}.
\end{align}
Consequently, the $\sigma_q$-singular operator is (Lemma~\ref{lemma:weighted-Q-SVD}):
\begin{align}
    Q_{\cE_p,\sigma_q}
    =
    \sum_{a} s_a(\cE_p,\sigma_q)^2F_a^2
    =
    \left(
        \frac{4(1-p)^2}{q}
        +
        \frac{p^2}{1-q}
    \right)\iden .
\end{align}
The optimal Gaussian converse bound in Theorem~\ref{theorem:weighted-gaussian-converse} is given by
\begin{equation}
    \inf_{0<q<1}
    \norm{Q_{\cE_p,\sigma_q}}_\infty
    =
    \inf_{0<q<1}
    \left[
        \frac{4(1-p)^2}{q}
        +
        \frac{p^2}{1-q}
    \right]
    =
    (2-p)^2,
\end{equation}
where the optimizer is $q_*= 2(1-p)/(2-p)$
with the endpoint cases $p=0,1$ understood by taking limits. Hence, Theorem~\ref{theorem:weighted-gaussian-converse} shows that for every $\lambda_1,\lambda_2>0$ with
$\lambda_1+\lambda_2<1$,
\begin{equation}\label{eq:erasure-weighted-gaussian}
    \limsup_{n\to\infty}
    \frac1n
    \log\log N_{(n,\lambda_1,\lambda_2)}(\cE_p)
    \leq
    2\log(2-p).
\end{equation}

Next, the weighted ellipsoid converse with $\sigma_q$-family (Proposition~\ref{prop:weighted-ellipsoid-converse}) shows
\begin{align}
    \limsup_{n\to\infty}
    \frac1n
    \log\log N_{(n,\lambda_1,\lambda_2)}(\cE_p)
    &\leq
    \inf_{0<q<1}
    \inf_{t>0}
    \log
    \left[
        s_0(\cE_p,\sigma_q)^t
        +
        3s_1(\cE_p,\sigma_q)^t
    \right] \\
    &=
    \inf_{0<q<1}
    \inf_{t>0}
    \log
    \left[
        \left(
            2\left[
                \frac{(1-p)^2}{q}
                +
                \frac{p^2}{1-q}
            \right]
        \right)^{t/2}
        +
        3
        \left(
            \frac{2(1-p)^2}{q}
        \right)^{t/2}
    \right]. 
    \label{eq:erasure-weighted-ellipsoid}
\end{align}
In particular, choosing the reference state $\sigma=\cE_p(\iden/2)
    =
    (1-p)\frac{\iden}{2}
    \oplus
    p\ketbra{e}$ gives
\begin{equation}\label{eq:erasure-weighted-ellipsoid-natural}
    \limsup_{n\to\infty}
    \frac1n
    \log\log N_{(n,\lambda_1,\lambda_2)}(\cE_p)
    \leq
    \inf_{t>0}
    \left[
        \frac t2
        +
        \log\left(
            1+3(1-p)^{t/2}
        \right)
    \right].
\end{equation}
Finally, the capacity bounds from Proposition~\ref{prop:cap-converse} show
\begin{align}
    \limsup_{n\to\infty}
    \frac{1}{n}\log\log N_{(n,\lambda_1,\lambda_2)}(\cE_p)
    &\leq
    1+C(\cE_p)
    = 1 + R(\cE_p) =
    2-p,
    \label{eq:CID(Ep)<=1+C}
\end{align}
where we used the identities $Q^{\dagger}(\cE_p)\leq R(\cE_p)$ and $C(\cE_p)= R(\cE_p) = 1-p$ \cite{Wilde2016, Tomamichel2017Qstrong-converse}.

\begin{figure}[ht]
    \centering
    \includegraphics[width=0.7\linewidth]{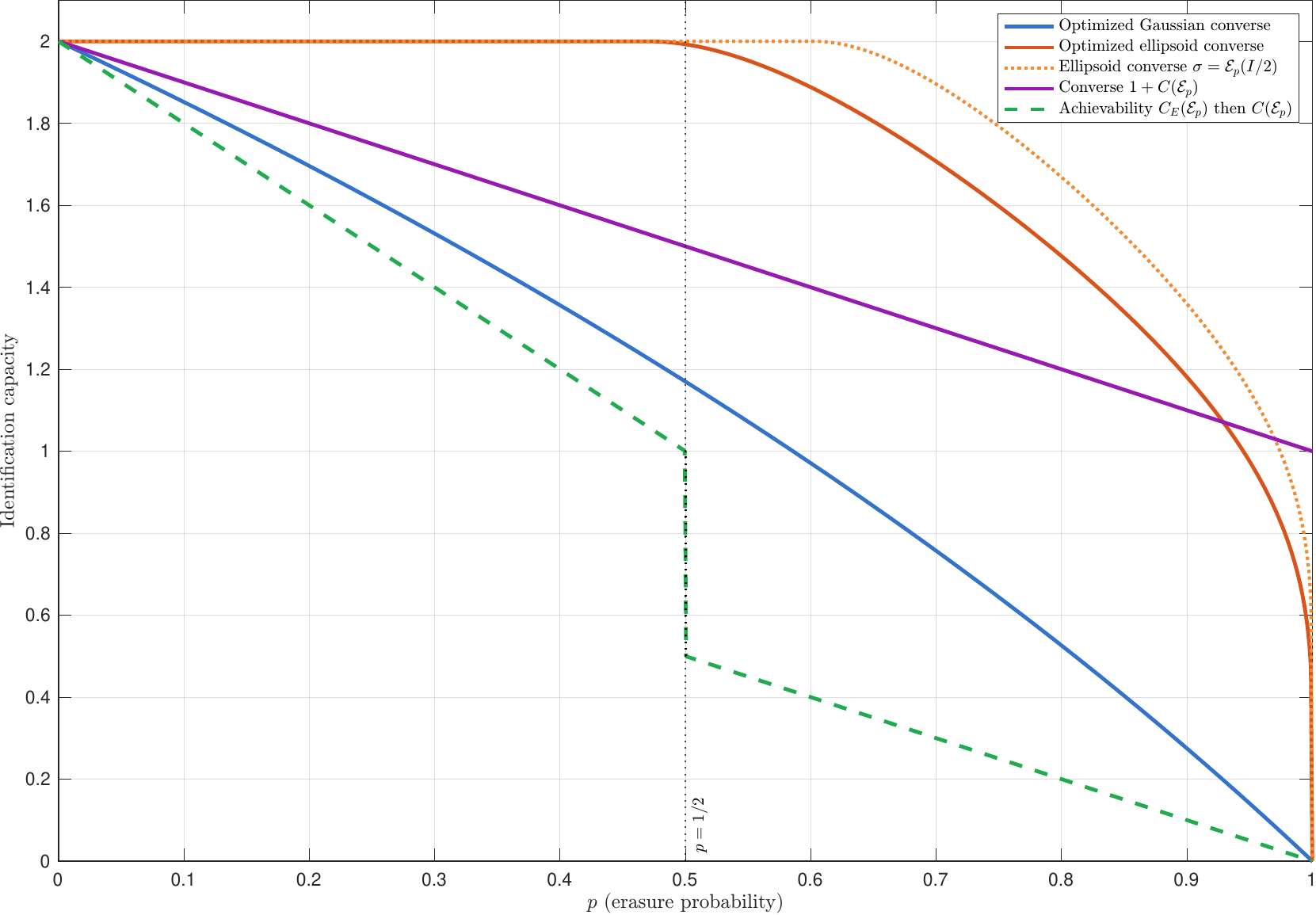}
    \caption{Strong converse and achievability bounds on the classical identification capacity of the qubit erasure channel $\cE_p$. The solid blue and orange curves show the Gaussian and ellipsoid strong converse bounds from Eqs.~\eqref{eq:erasure-weighted-gaussian} and \eqref{eq:erasure-weighted-ellipsoid}, respectively. The dashed orange curve is the ellipsoid bound with $\sigma=\cE_p(\iden /2)$ (Eq.~\eqref{eq:erasure-weighted-ellipsoid-natural}). The purple curve shows the capacity converse bound from Eq.~\eqref{eq:CID(Ep)<=1+C}. The green dashed curve shows the achievability bound from Eq.~\eqref{eq:qubit-erasure-achievability}.}
    \label{fig:qubit-erasure}
\end{figure}

For achievability, we combine the universal lower bound
$C_{\ID}(\cE_p)\geq C(\cE_p)$ with the stronger low-noise bound
$C_{\ID}(\cE_p)\geq C_E(\cE_p)$ from
\cite{Hayden2012QID-achievability} and
\cite[Remark 16 and Section 4]{Winter2013survey-ID}. Applying
$\id_A\otimes(\cE_p)_{A\to B}$ to a maximally entangled state
$\frac12\ketbra{\Omega}_{AA}$ gives
\begin{align}
    \rho_{AB}
    =
    \id_A\otimes(\cE_p)_{A\to B}
    \left(\frac12\ketbra{\Omega}_{AA}\right) =
    (1-p)\frac12\ketbra{\Omega}_{AB}
    +
    p\,\frac{\iden_A}{2}\otimes\ketbra{e}_B,
\end{align}
which has $
    I(A:B)_\rho
    =
    2(1-p)
    =
    C_E(\cE_p)$ \cite{Wilde2016} and $
    I(A\rangle B)_\rho
    =
    1-2p$. Thus, the coherent information is positive when $p<\frac12$, and we obtain
\begin{equation}
   C_{\ID}(\cE_p) \geq
    \begin{cases}
        C_E(\cE_p)=2(1-p), & 0\leq p<\dfrac12,\\[1mm]
        C(\cE_p)=1-p, & \dfrac12\leq p\leq 1.
    \end{cases}
    \label{eq:qubit-erasure-achievability}
\end{equation}
A plot of all the converse/achievability bounds on identification for $\cE_p$ is presented in Figure~\ref{fig:qubit-erasure}.

\subsubsection{Amplitude damping channel} \label{sec:amplitude}

Consider the qubit amplitude damping channel \(\cA_p:\B{\C^2}\to\B{\C^2}\) defined as
\begin{equation}
    \cA_p(X)=K_0XK_0^\dagger+K_1XK_1^\dagger, \qquad 0\leq p\leq1
\end{equation}
where
\begin{equation}
    K_0=\ketbra{0}+\sqrt{1-p}\ketbra{1},
    \qquad
    K_1=\sqrt p\,\ketbra{0}{1}.
\end{equation}

The channel is covariant under diagonal phase rotations \cite{Singh2021diagonal}, and by the
covariance/convexity twirling argument (Remark~\ref{remark:covariance-twirl-convex}), it suffices to
optimize over diagonal reference states
\begin{equation}
    \sigma_q:=q\ketbra0+(1-q)\ketbra1,
    \qquad 0<q<1.
\end{equation}
In the Pauli basis $F_0=\frac{\iden}{\sqrt{2}}, F_a = \frac{\sigma_a}{\sqrt{2}}$ for $a\in \{x,y,z \}$, the channel acts as follows:
\begin{equation}
    \cA_p(F_0)=\frac{\iden+p\sigma_z}{\sqrt2},
    \qquad
    \cA_p(F_z)=(1-p)F_z,
\end{equation}
\begin{equation}
    \cA_p(F_x)=\sqrt{1-p}F_x,
    \qquad
    \cA_p(F_y)=\sqrt{1-p}F_y,
\end{equation}
and the corresponding Gram matrix in $(\Bsa{\C^2}, \langle \cdot , \cdot \rangle_{\sigma_q})$ has the form
\begin{equation}\label{eq:amplitude-gram}
    \begin{pmatrix}
        C_{00} & 0 & 0 & C_{0z}\\
        0 & C_{xx} & 0 & 0\\
        0 & 0 & C_{yy} & 0\\
        C_{0z} & 0 & 0 & C_{zz}
    \end{pmatrix},
\end{equation}
where
\begin{alignat}{2}
    C_{00}
    &=
    \frac12
    \left(
        \frac{(1+p)^2}{q}
        +
        \frac{(1-p)^2}{1-q}
    \right),
    \qquad
    C_{zz}
    =
    \frac{(1-p)^2}{2q(1-q)}, \\
    C_{0z}
    &=
    \frac{1-p}{2}
    \left(
        \frac{1+p}{q}
        -
        \frac{1-p}{1-q}
    \right),
    \qquad
    C_{xx} = C_{yy}
    =
    \frac {1-p}{\sqrt{q(1-q)}} =: C_\perp.
\end{alignat}
Therefore, using Lemma~\ref{lemma:weighted-Q-SVD}, we get
\begin{equation}
    Q_{\cA_p,\sigma_q}
    =
    \frac{C_{00}+C_{zz}+2C_\perp}{2}\,\iden
    +
    C_{0z}\sigma_z,
\end{equation}
and
\begin{equation}
    \norm{Q_{\cA_p,\sigma_q}}_\infty
    =
    \frac{C_{00}+C_{zz}+2C_\perp}{2}
    +
    |C_{0z}| =: \Gamma_p(q)
\end{equation}
Thus the weighted Gaussian converse (Theorem~\ref{theorem:weighted-gaussian-converse}) gives
\begin{equation}\label{eq:amplitude-gaussian-optimized}
    \limsup_{n\to\infty}
    \frac1n
    \log\log N_{(n,\lambda_1,\lambda_2)}(\cA_p)
    \leq
    \inf_{0<q<1}
    \log \Gamma_p(q),
\end{equation}
This one-dimensional optimization can be evaluated explicitly. For simplicity, we only note the bound for the natural reference state $\sigma_{q_*}=\cA_p(\iden/2)$ with $ q_*=(1+p)/2$:
\begin{equation}\label{eq:amplitude-gaussian-natural}
    \limsup_{n\to\infty}
    \frac1n
    \log\log N_{(n,\lambda_1,\lambda_2)}(\cA_p)
    \leq
    2\log\left(
        1+\sqrt{\frac{1-p}{1+p}}
    \right).
\end{equation}
For the weighted ellipsoid converse, note that the \(\sigma_q\)-weighted singular values squared are given by the eigenvalues of the gram matrix \eqref{eq:amplitude-gram}: $\lambda_+(q), \lambda_-(q), C_\perp, C_\perp$, where
\begin{equation}
    \lambda_\pm(q)
    :=
    \frac{C_{00}+C_{zz}}{2}
    \pm
    \sqrt{
        \left(\frac{C_{00}-C_{zz}}{2}\right)^2
        +
        C_{0z}^2
    } .
\end{equation}
Therefore, Proposition~\ref{prop:weighted-ellipsoid-converse} shows
\begin{equation}\label{eq:amplitude-ellipsoid-optimized}
    \limsup_{n\to\infty}
    \frac1n
    \log\log N_{(n,\lambda_1,\lambda_2)}(\cA_p)
    \leq
    \inf_{0<q<1}\inf_{t>0}
    \log\left[
        \lambda_+(q)^{t/2}
        +
        \lambda_-(q)^{t/2}
        +
        2C_\perp^{t/2}
    \right].
\end{equation}
For the natural reference \(q_*=(1+p)/2\), this reduces to
\begin{equation}\label{eq:amplitude-ellipsoid-natural}
 \limsup_{n\to\infty}
    \frac1n
    \log\log N_{(n,\lambda_1,\lambda_2)}(\cA_p)
    \leq   \inf_{t>0}
    \left[
        \frac t2
        +
        2\log\left(
            1+
            \left(\frac{1-p}{1+p}\right)^{t/4}
        \right)
    \right].
\end{equation}
Finally, the capacity bounds from Proposition~\ref{prop:cap-converse} give
\begin{equation}\label{eq:amplitude-capacity-converse}
    \limsup_{n\to\infty}
    \frac1n
    \log\log N_{(n,\lambda_1,\lambda_2)}(\cA_p)
    \leq 1+Q^{\dagger}(\cA_p) \leq1+R_{\max}(\cA_p) = 1 + \log(2-p),
\end{equation}
where we used the bounds $Q^{\dagger}(\cA_p)\leq R(\cA_p) \leq R_{\max}(\cA_p)$ \cite{Tomamichel2017Qstrong-converse}. The \emph{max-Rains information} $R_{\max}(\cA_p) = \log (2-p)$ was computed in \cite[Proposition 2]{Rigovacca2018amplitude}. 

For achievability, we combine the universal lower bound
$C_{\ID}(\cA_p)\geq C(\cA_p)$ with the stronger low-noise bound
$C_{\ID}(\cA_p)\geq C_E(\cA_p)$ from
\cite{Hayden2012QID-achievability} and
\cite[Remark 16 and Section 4]{Winter2013survey-ID}.
For the input $(\sigma_q)_A:=q\ketbra0+(1-q)\ketbra1$ with purification $\ket{\psi_{q}}_{AA} = \sqrt{q}\ket{00} + \sqrt{1-q}\ket{11}$ and $(\rho_q)_{AB}= \id_A \otimes (\cA_p)_{A\to B} (\psi_q)$, the coherent and mutual information are
\begin{align}
    I(A\rangle B)_{\rho_q}
    &=
    h\bigl((1-p)(1-q)\bigr)-h(p(1-q)) \\
    I(A:B)_{\rho_q} &=
        h(q)+h\bigl((1-p)(1-q)\bigr)-h(p(1-q)).
\end{align}
When $0\leq p<\frac12$, the coherent information is positive for all $0<q<1$ \cite{Giovannetti2005amplitude}, and we obtain
\begin{equation}\label{eq:amplitude-achievability} 
   C_{\ID}(\cA_p) \geq
    \begin{cases}
        C_E(\cA_p), &0\leq p<\dfrac12,\\[1mm]
        C(\cA_p)\geq \chi(\cA_p)  \quad & \dfrac12\leq p\leq 1,
    \end{cases}
\end{equation}
where the entanglement-assisted classical capacity
\begin{align}
    C_E(\cA_p) &=\max_{0\leq q\leq 1}
    \left[
        h(q)+h\bigl((1-p)(1-q)\bigr)-h(p(1-q))
    \right], \quad \text{and} \\
    \chi(\cA_p)  &=\max_{0\leq q\leq 1}
    \left[
        h\bigl((1-p)(1-q)\bigr)
        -
        h\!\left(
            \frac{1+\sqrt{1-4p(1-p)(1-q)^2}}{2}
        \right)
    \right]
\end{align}
is the Holevo capacity \cite{Giovannetti2005amplitude}, which bounds the classical capacity from below: $C(\cA_p)\geq \chi (\cA_p)$ \cite{Wilde2016}. A plot of all the converse and achievability bounds for $\cA_p$ is presented in Figure~\ref{fig:qubit-amplitude}.

\begin{figure}[ht]
    \centering
    \includegraphics[width=0.7\linewidth]{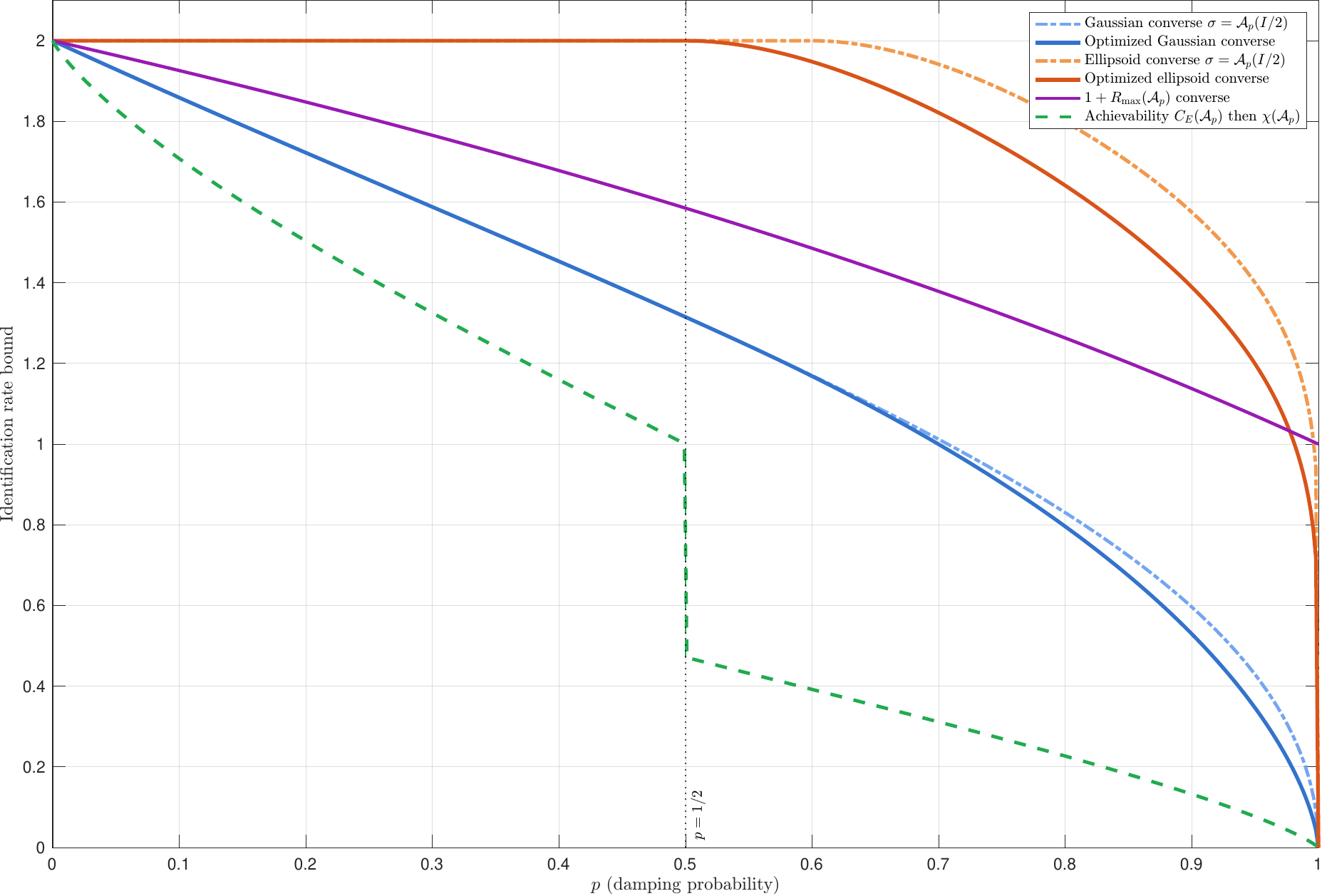}
    \caption{Strong converse and achievability bounds on the classical identification capacity of the qubit amplitude damping channel $\cA_p$. The solid blue and orange curves show the optimized Gaussian and Ellipsoid converse bounds from Eqs.~\eqref{eq:amplitude-gaussian-optimized} and \eqref{eq:amplitude-ellipsoid-optimized}, respectively. The dashed blue and orange curves show the $\sigma=\cA_p(\iden /2)$ Gaussian and Ellipsoid converse bounds from Eqs.~\eqref{eq:amplitude-gaussian-natural} and \eqref{eq:amplitude-ellipsoid-natural}, respectively. The purple curve shows the capacity converse bound from Eq.~\eqref{eq:amplitude-capacity-converse}. The green dashed curve shows the achievability bound from Eq.~\eqref{eq:amplitude-achievability}.}
    \label{fig:qubit-amplitude}
\end{figure}

\section{Discussion}
\label{sec:discussion}

\subsection{More general Euclidean geometries}\label{sec:general-geometries}

The weighted Euclidean structure on $\Bsa{B}$ for a given $\sigma\in \cD_+(B)$, defined as
\begin{equation}
    \langle Y,Z\rangle_\sigma:=\Tr(Y\sigma^{-1/2}Z \sigma^{-1/2}), \qquad \norm{Y}_{\sigma}^2:=\Tr(Y\sigma^{-1/2}Y \sigma^{-1/2}),
\end{equation}
ensures that
the weighted norm dominates the trace norm $\norm{\cdot}_1\leq \norm{\cdot}_{\sigma}$ (Lemma~\ref{lemma:weighted-norm-dominates-trace}). A more general formulation would be to use an arbitrary linear map $W:\Bsa B\to\Bsa B$ that is (Hilbert-Schmidt) self-adjoint and positive definite\footnote{Not to be confused with positivity-preserving maps, for which $Y\geq 0 \implies W(Y)\geq 0$. Here, we require $\Tr(Y W(Y))>0$ for all $\Bsa{B} \ni Y\neq 0$, which means $W$ is positive definite as an operator acting on $\Bsa{B}$.},
and define
\begin{equation}
    \langle Y,Z\rangle_W:=\Tr(YW(Z)), \qquad
    \norm{Y}_W^2:=\Tr(YW(Y)).
\end{equation}
We can equip \(\Bsa{B^{\otimes n}}\) with the product Euclidean structure induced by $W^{\otimes n}$, and define the \emph{trace norm domination cost}
\begin{equation}\label{eq:KW}
    K(W^{\otimes n})
    :=
    \sup_{Y\neq0}
    \frac{\norm{Y}_1}{\norm{Y}_{W^{\otimes n}}}.
\end{equation}
Equivalently, \(K(W^{\otimes n})\) is the smallest constant such that $\norm{Y}_1 \leq
    K(W^{\otimes n})\norm{Y}_{W^{\otimes n}}$ holds.  Let \(G_W\) be a standard Gaussian vector in
\((\Bsa B,\langle\cdot,\cdot\rangle_W)\), and define (c.f. Eq.~\eqref{eq:weighted-singular-Q}):
\begin{equation}\label{eq:Q_N,W}
    Q_{\cN,W}
    :=
    \E \left[\left(\cN^{*,W}(G_W)\right)^2 \right], 
\end{equation}
where \(\cN^{*,W}\) is the $W$-adjoint of $\cN: (\Bsa{A}, \langle \cdot , \cdot \rangle) \to ( \Bsa{B} , \langle \cdot , \cdot \rangle_W )$. Since the \(n\)-shot geometry is defined using
$W^{\otimes n}$, the corresponding singular operator factorizes like in Lemma~\ref{lemma:weighted-Q-multiplicativity}:
$Q_{\cN^{\otimes n},W^{\otimes n}}
    =
    Q_{\cN,W}^{\otimes n}.$ Then, the argument used in Theorem~\ref{theorem:weighted-gaussian-width-growth} shows
\begin{equation}
    w_{G,W^{\otimes n}}\bigl(I_n(\cN)\bigr)
    \leq
    \sqrt{2n\log d_A}\,
    \norm{Q_{\cN,W}}_\infty^{n/2},
\end{equation}
and combining this with Sudakov inequality and the domination cost \(K(W^{\otimes n})\) yields
\begin{equation}\label{eq:W-gaussian-converse-general}
    \limsup_{n\to\infty}
    \frac1n\log\log N_{(n,\lambda_1,\lambda_2)}(\cN)
    \leq
    \max\left\{
        0,\;
        2\kappa(W)+\log\norm{Q_{\cN,W}}_\infty
    \right\},
\end{equation}
exactly as was shown in Theorem~\ref{theorem:weighted-gaussian-converse}, where
\begin{equation}
    \kappa(W):=
    \limsup_{n\to\infty}
    \frac1n\log K(W^{\otimes n}).
\end{equation}
In the state-weighted case studied in this paper, $W_\sigma(\cdot)=\sigma^{-1/2}(\cdot)\sigma^{-1/2},$ for which Hölder's inequality (Lemma~\ref{lemma:weighted-norm-dominates-trace}) shows \(K(W^{\otimes n}_\sigma)=1\) for every \(n\). More general choices of $W$ might allow for better Euclidean structures
that are more suited to the geometry of the given channel image $I_n(\cN)$. The resulting optimization
problem:
\begin{equation}
    \inf_W
    \left[
        2\kappa(W)+\log\norm{Q_{\cN,W}}_\infty
    \right],
\end{equation}
over admissible Euclidean structures $W$, is a natural direction for future investigation.

\subsection{Correlated/Entangled weighted geometries}\label{sec:multi-letter-converse}
The converse bound developed in this paper uses \emph{product} Euclidean geometries on $\Bsa{B^{\otimes n}}$ defined by weighing maps $W_{\sigma^{\otimes n}}=W_{\sigma}^{\otimes n}$ (c.f. Eq.~\eqref{eq:Wsigma}) for $\sigma\in \cD_+(B)$, which results in a single-letter converse bound on identification via tensorization: $Q_{\mathcal N^{\otimes n},\sigma^{\otimes n}} = Q_{\mathcal N,\sigma}^{\otimes n}.$ For a given channel $\cN:\B{A}\to \B{B}$, denote this optimal converse bound from Theorem~\ref{theorem:weighted-gaussian-converse} by
\begin{equation}\label{eq:mu(N)}
    \mu(\cN)
    :=
    \inf_{\sigma\in\cD_+(B)} \norm{
        Q_{\cN,\sigma} }_\infty . 
\end{equation}

More generally, we can equip the \(n\)-shot output space with a weighted Euclidean structure $W_{\sigma_n}:\Bsa{B^{\otimes n}}\to \Bsa{B^{\otimes n}}$ associated with an arbitrary full-rank state
$\sigma_n\in\mathcal D_+(B^{\otimes n})$. Since $\|Y\|_1\leq \|Y\|_{\sigma_n}$ still holds according to Lemma~\ref{lemma:weighted-norm-dominates-trace}, the Gaussian width and covering argument of Theorem~\ref{theorem:weighted-gaussian-width-growth}, \ref{theorem:weighted-gaussian-converse} shows that for any channel $\cN:\B{A}\to \B{B}$: 
\begin{align}
    \limsup_{n\to\infty}
    \frac1n
    \log\log N_{(n,\lambda_1,\lambda_2)}(\cN)
    &\leq
    \limsup_{n\to \infty}
    \frac1n\log \mu (\cN^{\otimes n}) \\ 
    &= \inf_{n\in \N} \frac1n\log \mu (\cN^{\otimes n}) \leq \log \mu(\cN),
\end{align}
where the latter bound follows from Fekete's Lemma \cite{Fekete1923} and sub-multiplicativity: 
\begin{equation}
    \mu (\cN \otimes \cM) \leq \mu(\cN) \mu(\cM),
\end{equation}
which in turn follows by restricting the optimization in Eq.~\eqref{eq:mu(N)} to product states. In fact, $\mu(\cdot)$ can be strictly sub-multiplicative. This can already be seen for the qubit amplitude damping channel at \(p=1/2\), where (see Section~\ref{sec:amplitude})
\begin{equation}
    \mu(\cA_{1/2})
    =
    \left(1+\frac1{\sqrt3}\right)^2,
\end{equation}
attained at \(\sigma_*=\frac34\ketbra0+\frac14\ketbra1\). On the other hand, for 
\begin{equation}
    \sigma
    =
    \frac6{11}\ketbra{00}
    +
    \frac2{11}\ketbra{01}
    +
    \frac2{11}\ketbra{10}
    +
    \frac1{11}\ketbra{11} \in \cD_+(\C^2 \otimes \C^2),
\end{equation}
a direct computation gives
\begin{align}
 \mu(\cA_{1/2}^{\otimes 2}) \leq \left\|
        Q_{\cA_{1/2}^{\otimes2},\sigma}
    \right\|_\infty
    &=
    \frac{11\sqrt2}{16}
    +
    \frac{55\sqrt3}{48}
    +
    \frac{77}{24} \\
    &<\left(1+\frac1{\sqrt3}\right)^4 = \mu(\cA_{1/2})^2.
\end{align}
Hence, correlated multi-copy reference states can strictly improve the Gaussian converse in Theorem~\ref{theorem:weighted-gaussian-converse}. Understanding the multiplicativity of $\mu(\cdot)$ is therefore another interesting problem to study. Below, we note that for channels that are irreducibly covariant with respect to a compact group, $\mu(\cdot)$ is multiplicative.

\begin{lemma} \label{lemma:mu-multiplicative-irreducible-covariance}
Let $\cN:\B{A}\to\B{B}$ be covariant with respect to a compact group $G$:
\begin{equation}
 \forall  g\in G: \quad   \cN(U_g XU_g^\dagger)=V_g\cN(X)V_g^\dagger,
\end{equation}
where $g\mapsto U_g,V_g$ are unitary representations of $G$ and $g\mapsto V_g$ is irreducible. Then,
\begin{equation}
 \forall n\in \N: \quad   \mu(\cN^{\otimes n})=\mu(\cN)^n.
\end{equation}
\end{lemma}

\begin{proof}
By irreducibility of $g\mapsto V_g$, for any state $\sigma\in \cD(B)$, the output twirl with respect to the normalized Haar measure on $G$ satisfies
\begin{equation}
   \tau(\sigma) := \int_{g\in G} V_g \sigma V_g^\dagger\,dg= \tau := \frac{\iden_B}{d_B}.
\end{equation}
Hence, twirling locally over \(G^n\) maps every
\(\sigma\in\cD(B^{\otimes n})\) to \(\tau^{\otimes n}\). By covariance,
\begin{equation}
 \forall g=(g_1,\ldots,g_n)\in G^n: \quad   Q_{\cN^{\otimes n},\,V_g\sigma V_g^\dagger}
    =
    U_g
    Q_{\cN^{\otimes n},\sigma}
    U_g^\dagger,
\end{equation}
where \(V_g:=V_{g_1}\otimes\cdots\otimes V_{g_n}\) and
$U_g:=U_{g_1}\otimes\cdots\otimes U_{g_n}$ (see \eqref{eq:erasure-covariance-2},\eqref{eq:erasure-covariance-3}). Moreover,
\begin{equation}
    Q_{\cN^{\otimes n},\tau^{\otimes n}}
    \leq
    \int_{g\in G^n}
    U_g
    Q_{\cN^{\otimes n},\sigma}
    U_g^\dagger\,dg ,
\end{equation}
since $\sigma\mapsto Q_{\cN, \sigma}$ is convex (Lemma~\ref{lemma:weighted-Q-convex}).
Taking operator norms and using convexity and unitary invariance of
\(\|\cdot\|_\infty\), we get for every $\sigma\in \cD_+(B^{\otimes n})$,
\begin{equation}
    \left\|
        Q_{\cN^{\otimes n},\tau^{\otimes n}}
    \right\|_\infty
    \leq
    \left\|
        Q_{\cN^{\otimes n},\sigma}
    \right\|_\infty .
\end{equation}
Thus, \(\tau^{\otimes n}\) is optimal for $\mu(\cN^{\otimes n})$. Finally, since $Q_{\cN^{\otimes n},\tau^{\otimes n}}
    =
    Q_{\cN,\tau}^{\otimes n}.$ (Lemma~\ref{lemma:weighted-Q-multiplicativity}), we get 
\begin{equation}
    \mu(\cN^{\otimes n})
    =
    \left\|
        Q_{\cN,\tau}^{\otimes n}
    \right\|_\infty
    =
    \left\|
        Q_{\cN,\tau}
    \right\|_\infty^n
    =
    \mu(\cN)^n .
\end{equation}
\end{proof}

\subsection{Fully optimized Euclidean Gaussian width converse}\label{sec:full-opt-euclidean-converse}

Taking considerations from the previous two sections into account, we now state the strongest Gaussian width converse bound that can be obtained via arbitrary Euclidean structures. For a given channel $\cN:\B{A}\to \B{B}$, define
\begin{equation}
    \mu_*(\cN) := \inf_W \biggl( K(W)^2 \norm{Q_{\cN,W}}_\infty \biggr),
\end{equation}
where $K(\cdot)$ and $Q_{\cN,W}$ are the trace norm domination cost and the $W$-weighted singular operator, respectively (see Eq.~\eqref{eq:KW} and Eq.~\eqref{eq:Q_N,W}):
\begin{align}
    K(W) &:= \sup_{Y\neq 0} \frac{\norm{Y}_1}{\norm{Y}_W} \\
    Q_{\cN,W} &:=
    \E \left[\left(\cN^{*,W}(G_W)\right)^2 \right],
\end{align}
and the infimum is over all admissible Euclidean structures $W:\Bsa{B}\to \Bsa{B}$, i.e., over all (Hilbert-Schmidt) self-adjoint linear maps that are positive definite.

\begin{theorem}[Fully optimized Euclidean Gaussian converse]
\label{theorem:fully-optimized-euclidean-converse}
Let \(\cN:\B{A}\to\B{B}\) be a quantum channel. Then, for every \(\lambda_1,\lambda_2>0\) with
\(\lambda_1+\lambda_2<1\),
\begin{equation}
    \limsup_{n\to\infty}
    \frac1n
    \log\log N_{(n,\lambda_1,\lambda_2)}(\cN)
    \leq
    \limsup_{n\to\infty}
    \frac1n
    \log \mu_*(\cN^{\otimes n}).
\end{equation}
\end{theorem}

\begin{proof}
Fix \(n\) and a Euclidean structure \(W_n\) on
\(\Bsa{B^{\otimes n}}\). Consider a $(n,N,\lambda_1,\lambda_2)$ identification code $(\rho_i, D_i)_{i\in [N]}$ for $\cN$ and fix
\(0<\zeta<1-\lambda_1-\lambda_2\). If two code outputs
$\omega_i=\cN^{\otimes n}(\rho_i),\omega_j=\cN^{\otimes n}(\rho_j)\in I_n(\cN)$ lie in the same
\(\norm{\cdot}_{W_n}\)-ball of radius \(\zeta/K(W_n)\), then 
\begin{equation}
    \frac{1}{2}\norm{\omega_i-\omega_j}_1
    \leq
    \frac{1}{2}K(W_n) \norm{\omega_i-\omega_j}_{W_n}
    \leq
    \zeta
    <
    1-\lambda_1-\lambda_2,
\end{equation}
contradicting the identification separation condition (Eq.~\eqref{eq:CID-separation} and Lemma~\ref{lemma:CID-separation}). Hence, $N_{(n,\lambda_1,\lambda_2)}(\cN)
    \leq
    C_{\zeta/K(W_n)}
    \bigl(I_n(\cN);\norm{\cdot}_{W_n}\bigr).$
By Sudakov's inequality (Lemma~\ref{lemma:sudakov}) in the Euclidean space $(\Bsa{B^{\otimes n}}, \langle \cdot , \cdot \rangle_{W_n})$, there exists a universal constant \(C>0\) such that
\begin{equation}
    \log
    C_{\zeta/K(W_n)}
    \bigl(I_n(\cN);\norm{\cdot}_{W_n}\bigr)
    \leq
    C\, \frac{
    w_{G,W_n}\bigl(I_n(\cN)\bigr)^2 K(W_n)^2}{\zeta^2},
\end{equation}
and the matrix Gaussian series estimate (Lemma~\ref{lemma:gaussian-series-bound}, Remark~\ref{remark:weighted-width-upper}) shows
\begin{equation}
    w_{G,W_n}\bigl(I_n(\cN)\bigr)
    \leq
    \sqrt{2n\log d_A
    \norm{Q_{\cN^{\otimes n},W_n}}_\infty}.
\end{equation}
Therefore,
\begin{align}
    \log N_{(n,\lambda_1,\lambda_2)}(\cN) &\leq \log
    C_{\zeta/K(W_n)}
    \bigl(I_n(\cN);\norm{\cdot}_{W_n}\bigr) \\
    &\leq 
    2C \frac{ n\log d_A 
    \norm{Q_{\cN^{\otimes n},W_n}}_\infty K(W_n)^2}{\zeta^2} .
\end{align}
Taking infimum over \(W_n\), then taking logarithms once more, dividing by
\(n\), and letting \(n\to\infty\), the subexponential prefactor disappears and
yields the claim.
\end{proof}

\begin{remark}
The quantity \(K(W)^2\norm{Q_{\cN,W}}_\infty\) is invariant under rescaling of the
Euclidean structure. Indeed, replacing \(W\) by \(cW\), \(c>0\), gives
\begin{equation}
    K(cW)^2=c^{-1}K(W)^2,
    \qquad
    Q_{\cN,cW}=cQ_{\cN,W}.
\end{equation}
Thus, $K(cW)^2\norm{Q_{\cN,cW}}_\infty
    =
    K(W)^2\norm{Q_{\cN,W}}_\infty .$
\end{remark}

Finally, we note that a complete characterization of the classical identification capacity $C_{\ID}(\cN)$ and its corresponding strong converse for general channels $\cN:\B{A}\to \B{B}$ remains a pressing open problem \cite{Winter2013survey-ID}.

\section*{Acknowledgements}
I thank Michael Wolf for several illuminating discussions, and for suggesting the use of weighted Euclidean inner products on $\Bsa{B^{\otimes n}}$, which significantly improved the main results. I acknowledge support from the Deutsche Forschungsgemeinschaft (DFG, German Research Foundation) via TRR 352 – Project-ID 470903074. 

\section*{Appendices}
\addappheadtotoc
\begin{subappendices}

\renewcommand{\setthesubsection}{\Alph{subsection}}
\counterwithin{theorem}{subsection}
\renewcommand{\thetheorem}{\thesubsection.\arabic{theorem}}

\subsection{Technical lemmas}

\begin{lemma}[Matrix Gaussian series estimate on $\lambda_{\max}$] \label{lemma:gaussian-series-bound} \cite{Tropp2011gaussian-series-bound}, \cite[Theorem 4.6.1]{Tropp2015random-matrix-bounds} \\
For every finite family of Hermitian matrices $(A_k)_k$ of size $D\times D$ and i.i.d. standard real Gaussian random variables $(g_k)_{k}$, the following holds true:
\begin{equation}
    \E \lambda_{\max} \left( \sum_k g_k A_k \right)
    \leq
    \sqrt{2\log D
    \left\|
        \sum_k A_k^2
    \right\|_\infty}.
\end{equation}
\end{lemma}

\begin{lemma}[Gaussian Kahane--Khintchine inequality] \label{lemma:gaussian-banach-KK} \cite{Albiac2006analysis} \cite[Theorem 6.2.6]{hytonen2017analysis} \\  Let $E$ be a real or complex Banach space, and let $X=\sum_k g_k x_k$, where $(g_k)_{k\in [K]}$ are i.i.d.\ standard real Gaussian random variables and
$(x_k)_{k\in [K]}\subseteq E$ are fixed vectors. Then for every $1\leq p,q<\infty$, there exists a constant $K_{p,q}>0$, depending only on $p$ and $q$, such that
\begin{equation}
    (\E \|X\|^p)^{1/p}\le K_{p,q}\,(\E \|X\|^q)^{1/q}.
\end{equation}
\end{lemma}

\begin{lemma}[$L_1$--$L_2$ Gaussian Kahane--Khintchine for Hilbert spaces] 
\label{lemma:gaussian-L1-L2-complex}
Let $\mathcal H$ be a finite-dimensional complex Hilbert space, and let $\ket{\xi}=\sum_{k\in [K]} g_k \ket{v_k}\in \mathcal H$, where $(g_k)_{k\in [K]}$ are i.i.d.\ standard real Gaussian random variables and
$\{\ket{v_k}\}_{k\in [K]}\subseteq \mathcal H$ are fixed vectors. Then,
\begin{equation}
    \E \norm{\xi}
\geq
\sqrt{\frac{2}{\pi}}\,
\bigl(\E \norm{\xi}^2\bigr)^{1/2}.
\end{equation}
Moreover, the constant $\sqrt{2/\pi}$ above is optimal.
\end{lemma}

\begin{proof}
Let $\mathcal H_{\R}$ denote the underlying real Hilbert space of $\mathcal H$ with the inner product $\langle x,y\rangle_{\R}:=\Re \braket{x}{y}.$ The norm on $\mathcal H_{\R}$ agrees with the original complex norm on $\mathcal H$. Define the real covariance operator
\begin{equation}
\Gamma_{\R}
:=
\sum_{k} \ket{v_k}\!\bra{v_k}_{\R}, \quad \text{where} \quad \ket{v}\!\bra{w}_{\R}(x):=\langle w,x\rangle_{\R}\,\ket {v}. 
\end{equation}
Since $\Gamma_{\R}$ is positive semidefinite on $\mathcal H_{\R}$, there exists an orthonormal basis $
\ket{f_0}, \ldots , \ket{f_{r-1}}\in \mathcal H_{\R}
$ and eigenvalues $\lambda_0,\dots,\lambda_{r-1}\geq 0$ such that $\Gamma_{\R} \ket{f_j}=\lambda_j \ket{f_j}$.
Define
\begin{equation}
\ket{\xi'}:=\sum_{j\in [r]} \sqrt{\lambda_j}\, h_j \ket{f_j},
\end{equation}
where $h_0,\dots,h_{r-1}$ are i.i.d.\ standard real Gaussian random variables. We claim that $\ket{\xi}$ and $\ket{\xi'}$ have the same distribution as $\mathcal H_{\R}$-valued random vectors. Indeed, for every $\ket{x}\in \mathcal H_{\R}$,
$\langle x,\xi\rangle_{\R}
=
\sum_{k} g_k\,\langle x,v_k\rangle_{\R}$ is a centered real Gaussian random variable with variance
\begin{equation}
\E \langle x,\xi\rangle_{\R}^2
=
\sum_{k\in[K]} \langle x,v_k\rangle_{\R}^2
=
\langle x,\Gamma_{\R}x\rangle_{\R}.
\end{equation}
Similarly, $\langle x,\xi'\rangle_{\R}
=
\sum_{j} \sqrt{\lambda_j}\,h_j\,\langle x,f_j\rangle_{\R}$
is a centered real Gaussian with the same variance
\begin{equation}
\E \langle x,\xi'\rangle_{\R}^2
=
\sum_{j\in [r]} \lambda_j \langle x,f_j\rangle_{\R}^2
=
\langle x,\Gamma_{\R}x\rangle_{\R}.
\end{equation}
Hence, $\ket{\xi}$ and $\ket{\xi'}$ have the same characteristic function on $\mathcal H_{\R}$:
\begin{equation}
\E e^{i\langle x,\xi\rangle_{\R}}
=
\exp\!\left(-\frac12\langle x,\Gamma_{\R}x\rangle_{\R}\right)
=
\E e^{i\langle x,\xi'\rangle_{\R}}.
\end{equation}
Since the norms on $\mathcal H$ and $\mathcal H_{\R}$ agree, we have
\begin{equation}
\norm{\xi}
\stackrel{d}{=}
\norm{\xi'}
=
\left(\sum_{j\in [r]} \lambda_j h_j^2\right)^{1/2}.
\end{equation}
Set $s:=\sum_{j} \lambda_j = \E \norm{\xi}^2,
a_j:=\frac{\lambda_j}{s}.$ Then $a_j\geq 0$ and $\sum_j a_j=1$, and
\begin{equation}
\E \norm{\xi}
=
\sqrt{s}\,
\E \left(\sum_{j\in [r]} a_j h_j^2\right)^{1/2}.
\end{equation}
Define
\begin{equation}
F(a_0,\dots,a_{r-1})
:=
\E \left(\sum_{j\in [r]} a_j h_j^2\right)^{1/2}
\end{equation}
on the simplex
\begin{equation}
\Delta_r:=\left\{(a_0,\dots,a_{r-1})\in [0,1]^r:\ \sum_j a_j=1\right\}.
\end{equation}
For each fixed realization of $(h_j)_j$, the map
\begin{equation}
(a_0,\dots,a_{r-1})\longmapsto \left(\sum_{j} a_j h_j^2\right)^{1/2}
\end{equation}
is concave on $\Delta_r$, since it is the composition of a linear functional with the concave function $x\mapsto \sqrt{x}$ on $[0,\infty)$. Hence, $F$ is concave on $\Delta_r$ and its minimum is attained at an extreme point. Since $F$ is symmetric in the coordinates, all extreme points give the same value. So,
\begin{equation}
F(a_0,\dots,a_{r-1})\geq F(1,0,\dots,0)=\E |h_0|=\sqrt{\frac{2}{\pi}}.
\end{equation}
Consequently,
\begin{equation}
\E \norm{\xi}
\geq
\sqrt{s}\,\sqrt{\frac{2}{\pi}}
=
\sqrt{\frac{2}{\pi}}\,
\bigl(\E \norm{\xi}^2\bigr)^{1/2}.
\end{equation}

For optimality of $\sqrt{2/\pi}$, take $\ket{\xi}=g\ket{v}$ for any nonzero $\ket{v}\in \mathcal H$ with $g\sim N(0,1)$. Then, $\E \norm{\xi}
=
\E |g|\,\norm{v}
=
\sqrt{2/\pi} \norm{v}$ and $
\bigl(\E \norm{\xi}^2\bigr)^{1/2}
=
\norm{v},$
so equality holds.
\end{proof}

\begin{lemma} \label{lemma:depolarizing-optimization}
For $0< p< 1$, the following holds true
    \begin{align}
    R(p) &:=  \inf_{t>0}\left[
        \frac{t}{2}
        +
        \log\bigl(1+3(1-p)^t\bigr)
    \right] \\ 
    &= \begin{cases}
        2, & 0< p\leq 1-2^{-2/3},\\[1mm]
        2-D\!\left(\gamma(p)\middle\| \frac{3}{4}\right), & 1-2^{-2/3}<p<1,
    \end{cases}
    \end{align}
    where $\gamma(p):= -\frac{1}{2\log (1-p)}$ and $D(x \Vert y):= x\log \frac{x}{y} + (1-x)\log \frac{1-x}{1-y}$ is the binary relative entropy.
\end{lemma}

\begin{proof}
    For $0<p<1$, set $ q:=1-p\in (0,1)$ and
\begin{equation}
    \gamma= \gamma(p):= -\frac{1}{2\log q}= -\frac{1}{2\log(1-p)}.
\end{equation}
Let $u:=q^t$. Since $t=\frac{\log u}{\log q},$ we have $\frac{t}{2}=-\gamma\log u.$ Therefore\footnote{Since \(u\to 1\) as \(t\downarrow 0\) and the objective is continuous in $u$, we include 1 in the optimization interval.},
\begin{equation}
    R(p)
    =
    \inf_{0<u\leq 1}
    \left[
        -\gamma\log u+\log(1+3u)
    \right].
\end{equation}
Define
\begin{equation}
    f_\gamma(u):=-\gamma\log u+\log(1+3u),
    \qquad 0<u\leq 1.
\end{equation}
Differentiating gives
\begin{equation}
    f_\gamma'(u)
    =
    \frac{1}{\ln 2}
    \left(
        -\frac{\gamma}{u}+\frac{3}{1+3u}
    \right).
\end{equation}

We now distinguish two cases.

If $u_*>1$, equivalently $\gamma>\frac{3}{4}$, then $f_{\gamma}$ is non-increasing on $(0,1]$, and the minimum is attained at $u=1$, so $R(p)=f_\gamma(1)=\log 4=2.$ Now,
\begin{equation}
    \gamma>\frac{3}{4}
    \iff
    -\frac{1}{2\log(1-p)}>\frac{3}{4}
    \iff
    -\log(1-p)<\frac{2}{3}
    \iff
    p<1-2^{-2/3}.
\end{equation}
Thus $R(p)=2$ for $0\leq p< 1-2^{-2/3}.$

If $u_*\leq 1$, equivalently $\gamma\leq \frac{3}{4}$, then the minimum is attained at the critical point $u=u_*$:
\begin{equation}
    -\frac{\gamma}{u}+\frac{3}{1+3u}=0
    \qquad\Longleftrightarrow\qquad
    u_*=\frac{\gamma}{3(1-\gamma)}.
\end{equation} 
Since
\begin{equation}
    1+3u_*
    =
    1+\frac{\gamma}{1-\gamma}
    =
    \frac{1}{1-\gamma},
\end{equation}
we obtain
\begin{equation}
    R(p)
    =
    -\gamma\log\!\left(\frac{\gamma}{3(1-\gamma)}\right)
    +
    \log\!\left(\frac{1}{1-\gamma}\right).
\end{equation}
On the other hand,
\begin{equation}
    D\!\left(\gamma\middle\|\frac{3}{4}\right)
    =
    \gamma\log\frac{\gamma}{3/4}
    +
    (1-\gamma)\log\frac{1-\gamma}{1/4}
    =
    \gamma\log \gamma+(1-\gamma)\log(1-\gamma)-\gamma\log 3+2.
\end{equation}
Rearranging terms shows
\begin{equation}
    R(p)=
    2-D\!\left(\gamma (p)\middle\|\frac{3}{4}\right).
\end{equation}
\end{proof}

\subsection{Ellipsoid converse bound}\label{appen:ellipsoid}

In this section, we generalize the ellipsoid bound on the strong converse identification capacity of the qubit depolarizing channel from \cite[Theorem 5]{ye2026strongconverseboundsclassical} to \emph{all} quantum channels. We will make use of the following result from \cite{Dumer2006ellipse} about the covering number of ellipsoids.

\begin{theorem} \cite{Dumer2006ellipse} \label{theorem:dumer}
    Let $s=(s_0, \ldots ,s_{m-1})\in \R^m$ with $s_i>0$ for each $i\in [m]$ be given. Define the \emph{ellipsoid} with semi-axes lengths given by $s$ as follows:
    \begin{equation}
        \cE := \left\{ x\in \R^m : \sum_i \left(\frac{x_i}{s_i}\right)^2 \leq 1 \right\} \subseteq \R^m.
    \end{equation} 
    Let $\theta\in(0,1/2)$, $\mu_{\theta}:= |\{i\in [m] : s_i^2 >1 - \theta \}| $ and $K:= \sum_{i: s_i>1} \ln s_i$. Then, 
    \begin{equation}
        \ln C_{1}(\cE; d) \leq K + \mu_{\theta} \ln (3/\theta),
    \end{equation}
    where $d(x,y):=\sqrt{\sum_i (x_i-y_i)^2}$ is the standard Euclidean distance.
\end{theorem}

\begin{proposition}\label{prop:weighted-ellipsoid-converse-appen}
Let $\cN:\B{A}\to\B{B}$ be a quantum channel. For a full rank state $\sigma\in\cD_+(B)$, let $s_a(\cN,\sigma)$
denote the $\sigma$-weighted singular values of $\cN$.
Then, for every $\lambda_1,\lambda_2>0$ with $\lambda_1+\lambda_2<1$,
\begin{equation}
    \limsup_{n\to\infty}
    \frac1n\log\log N_{(n,\lambda_1,\lambda_2)}(\cN)
    \le
    \inf_{\sigma\in \cD_+(B)}
    \inf_{t>0}
    \log\left(\sum_a s_a(\cN,\sigma)^t\right).
\end{equation}
\end{proposition}

\begin{proof}
Let $\cN(F_a)=s_aG_a$
be a $\sigma$-weighted singular value decomposition, where $s_a=s_a(\cN,\sigma)$, $\{F_a\}_a\subseteq\Bsa{A}$
is Hilbert--Schmidt orthonormal and $\{G_a\}_{a:s_a>0}\subseteq\Bsa{B}$ is
$\sigma$-orthonormal. For a multi-index
$\alpha=(\alpha_1,\ldots,\alpha_n)\in[d_A^2]^n$, set
\begin{equation}
    s_\alpha:=\prod_{j=1}^n s_{\alpha_j},
    \qquad
    F_\alpha:=F_{\alpha_1}\otimes\cdots\otimes F_{\alpha_n},
    \qquad
    G_\alpha:=G_{\alpha_1}\otimes\cdots\otimes G_{\alpha_n}.
\end{equation}
Then, $\cN^{\otimes n}(F_\alpha)=s_\alpha G_\alpha,$
and $\{G_\alpha\}_{\alpha:s_\alpha>0}$ is
$\sigma^{\otimes n}$-orthonormal. 

Since every density matrix $\rho\in\cD(A^{\otimes n})$ satisfies $\norm{\rho}_2^2=\Tr\rho^2\leq \Tr\rho=1,$ the output image $I_n(\cN)$ is contained in the weighted ellipsoid
\begin{equation}
  I_n(\cN) \subseteq  E_n:=
    \left\{
        \sum_{\alpha:s_\alpha>0} y_\alpha G_\alpha:
        \sum_{\alpha:s_\alpha>0}
        \left(\frac{y_\alpha}{s_\alpha}\right)^2
        \leq 1
    \right\}
    \subseteq
    \bigl(\Bsa{B^{\otimes n}},\norm{\cdot}_{\sigma^{\otimes n}}\bigr).
\end{equation}
Now, consider a $(n,N,\lambda_1,\lambda_2)$ identification code $(\rho_i, D_i)_{i\in [N]}$ for $\cN$ (c.f. Definition~\ref{def:IDcodes}) and fix $0<\zeta<1-\lambda_1-\lambda_2$. If two code outputs $\omega_i,\omega_j \in I_n(\cN)$,
$\omega_i=\cN^{\otimes n}(\rho_i)$, lie in the same
$\norm{\cdot}_{\sigma^{\otimes n}}$-ball of radius $\zeta$, then
\begin{equation}
    \frac{1}{2}\norm{\omega_i-\omega_j}_1
    \leq
    \frac{1}{2}\norm{\omega_i-\omega_j}_{\sigma^{\otimes n}}
    \leq
    \zeta
    <
    1-\lambda_1-\lambda_2,
\end{equation}
where we used $\norm{Y}_1\leq\norm{Y}_{\sigma^{\otimes n}}$ (Lemma~\ref{lemma:weighted-norm-dominates-trace}). This contradicts
Eq.~\eqref{eq:CID-separation} and Lemma~\ref{lemma:CID-separation}. Hence, any
$\zeta$-cover of $E_n$ gives an upper bound on the code size:
\begin{equation}
    N_{(n,\lambda_1,\lambda_2)}(\cN)
    \leq
    C_\zeta(E_n;\norm{\cdot}_{\sigma^{\otimes n}}) = C_1(\widetilde E_n ; \norm{\cdot}_{\sigma^{\otimes n}}),
\end{equation}
where the rescaled ellipsoid $\widetilde E_n:=\zeta^{-1}E_n$ has semiaxes lengths $\widetilde s_\alpha:=\zeta^{-1}s_\alpha.$
Applying Dumer's ellipsoid covering estimate (Theorem~\ref{theorem:dumer}), for any fixed
$\theta\in(0,1/2)$,
\begin{equation}
    \ln C_1(\widetilde E_n)
    \leq
    K_n+\mu_{n,\theta}\ln(3/\theta),
\end{equation}
where $K_n:=\sum_{\alpha:\widetilde s_\alpha>1}\ln\widetilde s_\alpha$ and $ \mu_{n,\theta}:=
    \left|
        \left\{
            \alpha:
            \widetilde s_\alpha^2>1-\theta
        \right\}
    \right|.$
Now, fix $t>0$. Since
\begin{align}
    \mu_{n,\theta} &\leq (1-\theta)^{-t/2}\sum_{\alpha: s_{\alpha}>0} \widetilde{s}_\alpha^t, \\
    K_n &\leq \frac{1}{t}\sum_{\alpha: s_{\alpha}>0} \widetilde{s}_\alpha^t,
\end{align}
where we used the simple estimate $\ln x \leq x^t/t$ for $x>0$, we can write
\begin{equation}
    \ln C_{1}\bigl(\widetilde E_n;\norm{\cdot}_{\sigma^{\otimes n}}\bigr)
    \leq
    \left(
        \frac{1}{t} + (1-\theta)^{-t/2}\ln(3/\theta)
    \right)
    \sum_{\alpha: s_{\alpha}>0} \widetilde{s}_\alpha^t.
\end{equation}
Finally, since
\begin{equation}
    \sum_{\alpha:s_\alpha>0}\widetilde s_\alpha^t
    =
    \zeta^{-t}
    \sum_{\alpha:s_\alpha>0}s_\alpha^t
    =
    \zeta^{-t}
    \left(\sum_a s_a^t\right)^n,
\end{equation}
we get
\begin{align}
 \limsup_{n\to\infty}  \frac{1}{n} \log\log N_{(n,\lambda_1, \lambda_2)}(\cN) &\leq \limsup_{n\to\infty}\frac{1}{n} \log\log C_{1}\bigl(\widetilde{E}_n;\norm{\placeholder}_{\sigma^{\otimes n}}\bigr) \\
    &\leq \log\left(\sum_a s_a^t\right).
\end{align}
Since this holds for every $t>0$ and $\sigma\in \cD_+(B)$, we obtain the desired claim.
\end{proof}

\end{subappendices}

\bibliographystyle{plainurl}
\bibliography{references}

\begin{thebibliography}{10}

\bibitem{Ahlswede1989ID}
R.~Ahlswede and G.~Dueck.
\newblock Identification via channels.
\newblock {\em IEEE Transactions on Information Theory}, 35(1):15–29, 1989.
\newblock URL: \url{http://dx.doi.org/10.1109/18.42172}, \href {https://doi.org/10.1109/18.42172} {\path{doi:10.1109/18.42172}}.

\bibitem{Ahlswede2002strong-ID}
R.~Ahlswede and A.~Winter.
\newblock Strong converse for identification via quantum channels.
\newblock {\em IEEE Transactions on Information Theory}, 48(3):569–579, March 2002.
\newblock URL: \url{http://dx.doi.org/10.1109/18.985947}, \href {https://doi.org/10.1109/18.985947} {\path{doi:10.1109/18.985947}}.

\bibitem{Ando1979convex}
T.~Ando.
\newblock Concavity of certain maps on positive definite matrices and applications to hadamard products.
\newblock {\em Linear Algebra and its Applications}, 26:203–241, August 1979.
\newblock URL: \url{http://dx.doi.org/10.1016/0024-3795(79)90179-4}, \href {https://doi.org/10.1016/0024-3795(79)90179-4} {\path{doi:10.1016/0024-3795(79)90179-4}}.

\bibitem{Avidan2015analysis}
Shiri Artstein-Avidan, Apostolos Giannopoulos, and Vitali Milman.
\newblock {\em Asymptotic Geometric Analysis, Part I}.
\newblock American Mathematical Society, June 2015.
\newblock URL: \url{http://dx.doi.org/10.1090/surv/202}, \href {https://doi.org/10.1090/surv/202} {\path{doi:10.1090/surv/202}}.

\bibitem{Atif2024CIDstrongconverse}
Touheed~Anwar Atif, S.~Sandeep Pradhan, and Andreas Winter.
\newblock Quantum soft-covering lemma with applications to rate-distortion coding, resolvability and identification via quantum channels.
\newblock {\em International Journal of Quantum Information}, 22(05), 2024.
\newblock URL: \url{http://dx.doi.org/10.1142/S0219749924400136}, \href {https://doi.org/10.1142/s0219749924400136} {\path{doi:10.1142/s0219749924400136}}.

\bibitem{BethRuskai2002qubit}
Mary Beth~Ruskai, Stanislaw Szarek, and Elisabeth Werner.
\newblock An analysis of completely-positive trace-preserving maps on ${M}_2$.
\newblock {\em Linear Algebra and its Applications}, 347(1-3):159–187, May 2002.
\newblock URL: \url{http://dx.doi.org/10.1016/S0024-3795(01)00547-X}, \href {https://doi.org/10.1016/s0024-3795(01)00547-x} {\path{doi:10.1016/s0024-3795(01)00547-x}}.

\bibitem{bhatia2015positive}
R.~Bhatia.
\newblock {\em Positive Definite Matrices}.
\newblock Princeton Series in Applied Mathematics. Princeton University Press, 2015.
\newblock URL: \url{https://books.google.co.in/books?id=Y22YDwAAQBAJ}.

\bibitem{Bhatia1997matrix}
Rajendra Bhatia.
\newblock {\em Matrix Analysis}.
\newblock Springer New York, 1997.
\newblock URL: \url{http://dx.doi.org/10.1007/978-1-4612-0653-8}, \href {https://doi.org/10.1007/978-1-4612-0653-8} {\path{doi:10.1007/978-1-4612-0653-8}}.

\bibitem{Cover2005book}
Thomas~M. Cover and Joy~A. Thomas.
\newblock {\em Elements of Information Theory}.
\newblock Wiley, April 2005.
\newblock URL: \url{http://dx.doi.org/10.1002/047174882X}, \href {https://doi.org/10.1002/047174882x} {\path{doi:10.1002/047174882x}}.

\bibitem{Dudley1967}
R.M Dudley.
\newblock The sizes of compact subsets of {H}ilbert space and continuity of {G}aussian processes.
\newblock {\em Journal of Functional Analysis}, 1(3):290–330, October 1967.
\newblock URL: \url{http://dx.doi.org/10.1016/0022-1236(67)90017-1}, \href {https://doi.org/10.1016/0022-1236(67)90017-1} {\path{doi:10.1016/0022-1236(67)90017-1}}.

\bibitem{Dumer2006ellipse}
Ilya Dumer.
\newblock Covering an ellipsoid with equal balls.
\newblock {\em Journal of Combinatorial Theory, Series A}, 113(8):1667–1676, November 2006.
\newblock URL: \url{http://dx.doi.org/10.1016/j.jcta.2006.03.021}, \href {https://doi.org/10.1016/j.jcta.2006.03.021} {\path{doi:10.1016/j.jcta.2006.03.021}}.

\bibitem{Fekete1923}
M.~Fekete.
\newblock {\"U}ber die verteilung der wurzeln bei gewissen algebraischen gleichungen mit ganzzahligen koeffizienten.
\newblock {\em Mathematische Zeitschrift}, 17(1):228–249, December 1923.
\newblock URL: \url{http://dx.doi.org/10.1007/BF01504345}, \href {https://doi.org/10.1007/bf01504345} {\path{doi:10.1007/bf01504345}}.

\bibitem{Albiac2006analysis}
Nigel J.~Kalton Fernando~Albiac.
\newblock {\em Topics in Banach Space Theory}.
\newblock Springer-Verlag, 2006.
\newblock URL: \url{http://dx.doi.org/10.1007/0-387-28142-8}, \href {https://doi.org/10.1007/0-387-28142-8} {\path{doi:10.1007/0-387-28142-8}}.

\bibitem{Giovannetti2005amplitude}
Vittorio Giovannetti and Rosario Fazio.
\newblock Information-capacity description of spin-chain correlations.
\newblock {\em Physical Review A}, 71(3), March 2005.
\newblock URL: \url{http://dx.doi.org/10.1103/PhysRevA.71.032314}, \href {https://doi.org/10.1103/physreva.71.032314} {\path{doi:10.1103/physreva.71.032314}}.

\bibitem{Hayden2012QID-achievability}
Patrick Hayden and Andreas Winter.
\newblock Weak decoupling duality and quantum identification.
\newblock {\em IEEE Transactions on Information Theory}, 58(7):4914–4929, July 2012.
\newblock URL: \url{http://dx.doi.org/10.1109/TIT.2012.2191695}, \href {https://doi.org/10.1109/tit.2012.2191695} {\path{doi:10.1109/tit.2012.2191695}}.

\bibitem{hytonen2017analysis}
Tuomas Hyt\"{o}nen, Jan van Neerven, Mark Veraar, and Lutz Weis.
\newblock {\em Analysis in Banach Spaces}.
\newblock Springer International Publishing, 2017.
\newblock URL: \url{http://dx.doi.org/10.1007/978-3-319-69808-3}, \href {https://doi.org/10.1007/978-3-319-69808-3} {\path{doi:10.1007/978-3-319-69808-3}}.

\bibitem{King2003depolarizing}
C.~King.
\newblock The capacity of the quantum depolarizing channel.
\newblock {\em IEEE Transactions on Information Theory}, 49(1):221–229, January 2003.
\newblock URL: \url{http://dx.doi.org/10.1109/TIT.2002.806153}, \href {https://doi.org/10.1109/tit.2002.806153} {\path{doi:10.1109/tit.2002.806153}}.

\bibitem{King2001unital-qubit-C}
C.~King and M.B. Ruskai.
\newblock Minimal entropy of states emerging from noisy quantum channels.
\newblock {\em IEEE Transactions on Information Theory}, 47(1):192–209, 2001.
\newblock URL: \url{http://dx.doi.org/10.1109/18.904522}, \href {https://doi.org/10.1109/18.904522} {\path{doi:10.1109/18.904522}}.

\bibitem{King2002additive}
Christopher King.
\newblock Additivity for unital qubit channels.
\newblock {\em Journal of Mathematical Physics}, 43(10):4641–4653, October 2002.
\newblock URL: \url{http://dx.doi.org/10.1063/1.1500791}, \href {https://doi.org/10.1063/1.1500791} {\path{doi:10.1063/1.1500791}}.

\bibitem{Kubo1980mean}
Fumio Kubo and Tsuyoshi Ando.
\newblock Means of positive linear operators.
\newblock {\em Mathematische Annalen}, 246(3):205–224, October 1980.
\newblock URL: \url{http://dx.doi.org/10.1007/BF01371042}, \href {https://doi.org/10.1007/bf01371042} {\path{doi:10.1007/bf01371042}}.

\bibitem{Ledoux1991prob-banach}
Michel Ledoux and Michel Talagrand.
\newblock {\em Probability in Banach Spaces}.
\newblock Springer Berlin Heidelberg, 1991.
\newblock URL: \url{http://dx.doi.org/10.1007/978-3-642-20212-4}, \href {https://doi.org/10.1007/978-3-642-20212-4} {\path{doi:10.1007/978-3-642-20212-4}}.

\bibitem{Lober1999thesis-ID}
Peter L\"ober.
\newblock {\em Quantum channels and simultaneous ID coding}.
\newblock PhD thesis, Bielefeld University., 1999.
\newblock URL: \url{https://pub.uni-bielefeld.de/record/2303327}.

\bibitem{Miranowicz2008E_R-R}
Adam Miranowicz and Satoshi Ishizaka.
\newblock Closed formula for the relative entropy of entanglement.
\newblock {\em Physical Review A}, 78(3), 2008.
\newblock URL: \url{http://dx.doi.org/10.1103/PhysRevA.78.032310}, \href {https://doi.org/10.1103/physreva.78.032310} {\path{doi:10.1103/physreva.78.032310}}.

\bibitem{Rains1999}
E.~M. Rains.
\newblock Bound on distillable entanglement.
\newblock {\em Physical Review A}, 60(1):179–184, 1999.
\newblock URL: \url{http://dx.doi.org/10.1103/PhysRevA.60.179}, \href {https://doi.org/10.1103/physreva.60.179} {\path{doi:10.1103/physreva.60.179}}.

\bibitem{Rigovacca2018amplitude}
Luca Rigovacca, Go~Kato, Stefan B\"{a}uml, M~S Kim, W~J Munro, and Koji Azuma.
\newblock Versatile relative entropy bounds for quantum networks.
\newblock {\em New Journal of Physics}, 20(1):013033, January 2018.
\newblock URL: \url{http://dx.doi.org/10.1088/1367-2630/aa9fcf}, \href {https://doi.org/10.1088/1367-2630/aa9fcf} {\path{doi:10.1088/1367-2630/aa9fcf}}.

\bibitem{Shannon1948}
C.~E. Shannon.
\newblock A mathematical theory of communication.
\newblock {\em Bell System Technical Journal}, 27(3):379–423, July 1948.
\newblock URL: \url{http://dx.doi.org/10.1002/j.1538-7305.1948.tb01338.x}, \href {https://doi.org/10.1002/j.1538-7305.1948.tb01338.x} {\path{doi:10.1002/j.1538-7305.1948.tb01338.x}}.

\bibitem{Singh2022detecting}
Satvik Singh and Nilanjana Datta.
\newblock Detecting positive quantum capacities of quantum channels.
\newblock {\em npj Quantum Information}, 8(1), May 2022.
\newblock URL: \url{http://dx.doi.org/10.1038/s41534-022-00550-2}, \href {https://doi.org/10.1038/s41534-022-00550-2} {\path{doi:10.1038/s41534-022-00550-2}}.

\bibitem{Singh2021diagonal}
Satvik Singh and Ion Nechita.
\newblock Diagonal unitary and orthogonal symmetries in quantum theory.
\newblock {\em Quantum}, 5:519, August 2021.
\newblock URL: \url{http://dx.doi.org/10.22331/q-2021-08-09-519}, \href {https://doi.org/10.22331/q-2021-08-09-519} {\path{doi:10.22331/q-2021-08-09-519}}.

\bibitem{Sudakov1971}
V.~N. Sudakov.
\newblock Gaussian random processes and measures of solid angles in hilbert space.
\newblock {\em Dokl. Akad. Nauk SSSR}, 197(1):43--45, 1971.
\newblock URL: \url{https://mathscinet.ams.org/mathscinet-getitem?mr=0288832}.

\bibitem{Tomamichel2017Qstrong-converse}
Marco Tomamichel, Mark~M. Wilde, and Andreas Winter.
\newblock Strong converse rates for quantum communication.
\newblock {\em IEEE Transactions on Information Theory}, 63(1):715–727, January 2017.
\newblock URL: \url{http://dx.doi.org/10.1109/TIT.2016.2615847}, \href {https://doi.org/10.1109/tit.2016.2615847} {\path{doi:10.1109/tit.2016.2615847}}.

\bibitem{Tropp2011gaussian-series-bound}
Joel~A. Tropp.
\newblock User-friendly tail bounds for sums of random matrices.
\newblock {\em Foundations of Computational Mathematics}, 12(4):389–434, August 2011.
\newblock URL: \url{http://dx.doi.org/10.1007/s10208-011-9099-z}, \href {https://doi.org/10.1007/s10208-011-9099-z} {\path{doi:10.1007/s10208-011-9099-z}}.

\bibitem{Tropp2015random-matrix-bounds}
Joel~A. Tropp.
\newblock An introduction to matrix concentration inequalities.
\newblock {\em Foundations and Trends{\textregistered} in Machine Learning}, 8(1–2):1–230, May 2015.
\newblock URL: \url{http://dx.doi.org/10.1561/2200000048}, \href {https://doi.org/10.1561/2200000048} {\path{doi:10.1561/2200000048}}.

\bibitem{Vedral1997E_R}
V.~Vedral, M.~B. Plenio, M.~A. Rippin, and P.~L. Knight.
\newblock Quantifying entanglement.
\newblock {\em Physical Review Letters}, 78(12):2275–2279, March 1997.
\newblock URL: \url{http://dx.doi.org/10.1103/PhysRevLett.78.2275}, \href {https://doi.org/10.1103/physrevlett.78.2275} {\path{doi:10.1103/physrevlett.78.2275}}.

\bibitem{Vershynin2018HDP}
Roman Vershynin.
\newblock {\em High-Dimensional Probability: An Introduction with Applications in Data Science}.
\newblock Cambridge University Press, September 2018.
\newblock URL: \url{http://dx.doi.org/10.1017/9781108231596}, \href {https://doi.org/10.1017/9781108231596} {\path{doi:10.1017/9781108231596}}.

\bibitem{Wilde2016}
Mark~M. Wilde.
\newblock {\em Quantum Information Theory}.
\newblock Cambridge University Press, November 2016.
\newblock URL: \url{http://dx.doi.org/10.1017/9781316809976}, \href {https://doi.org/10.1017/9781316809976} {\path{doi:10.1017/9781316809976}}.

\bibitem{Winter2013survey-ID}
Andreas Winter.
\newblock {\em Identification via Quantum Channels}, page 217–233.
\newblock Springer Berlin Heidelberg, 2013.
\newblock URL: \url{http://dx.doi.org/10.1007/978-3-642-36899-8_9}, \href {https://doi.org/10.1007/978-3-642-36899-8_9} {\path{doi:10.1007/978-3-642-36899-8_9}}.

\bibitem{Wolf2012Qtour}
M.~M. Wolf.
\newblock Quantum channels and operations: Guided tour.
\newblock {\em (unpublished)}, 2012.
\newblock URL: \url{https://mediatum.ub.tum.de/node?id=1701036}.

\bibitem{ye2026strongconverseboundsclassical}
Liuhang Ye, Bjarne Bergh, and Nilanjana Datta.
\newblock Strong converse bounds on the classical identification capacity of the qubit depolarizing channel, 2026.
\newblock URL: \url{https://arxiv.org/abs/2603.29987}, \href {https://arxiv.org/abs/2603.29987} {\path{arXiv:2603.29987}}.

\end{thebibliography}

\end{document}